\documentclass[manuscript,nonacm]{acmart}
\usepackage[nameinlink,capitalise]{cleveref}
\usepackage{xspace}
\usepackage{amsmath,dsfont,amsthm}
\usepackage[shortlabels]{enumitem}
\usepackage{tikz}

\setcopyright{none}
\copyrightyear{2021}
\acmYear{2021}
\acmDOI{}
\acmConference[]{}{}{}
\acmBooktitle{}
\acmPrice{}
\acmISBN{}
 
 \def\E{\mathbf{E}}
\renewcommand{\Pr}[1]{\mathds{P}\left(#1\right) }
\def\Prob{\mathds{P}}
\newcommand{\Ex}[1]{\mathds{E}\left(#1\right) }

\newcommand{\Ind}{\mathds{1}}
\newcommand{\ind}{\Ind}
\newcommand{\given}{\middle|}

\newcommand{\ignore}[1]{}

 \newcommand{\var}{\mathds{V}ar}

\newcommand{\minority}{\textsc{Diversification}\xspace}

\newcommand{\eps}{\varepsilon}

\begin{document}

\title{Diversity, Fairness, and Sustainability in Population Protocols}
\author{Nan Kang}
\affiliation{%
\institution{King's College London}
 \country{UK}   
 }
\author{Frederik Mallmann-Trenn }
\affiliation{%
\institution{King's College London}
 \country{UK}   
 }
\author{Nicol\'as Rivera}
\affiliation{%
\institution{IMDF, and Universidad de Valparaíso}
 \country{Chile}   
 }


\begin{abstract}
   Over the years, population protocols with the goal of reaching consensus have been studied in great depth. However, many systems in the real-world do not result in all agents eventually reaching consensus, but rather in the opposite: they converge to a state of rich diversity. 
    Consider for example task allocation in ants. If eventually all ants perform the same task, then the colony will perish (lack of food, no brood care, etc.). Then, it is vital for the survival of the colony to have a diverse set of tasks and enough ants working on each task. What complicates matters is that ants need to switch tasks periodically to adjust the needs of the colony; e.g., when too many foragers fell victim to other ant colonies. A further difficulty is that not all tasks are equally important and maybe they need to keep certain proportions in the distribution of the task. How can ants keep a healthy and balanced allocation of tasks?
    
    To answer this question, we propose a simple population protocol for $n$ agents on a complete graph and an arbitrary initial distribution of $k$ colours (tasks). In this protocol we assume that each colour $i$ has an associated weight (importance or value) $w_i \geq 1$. By denoting $w$ as the sum of the weights of different colours, we show that the protocol converges in $O(w^2 n \log n)$ rounds to a configuration where the number of agents supporting each colour $i$ is concentrated on the fair share $w_in/w$ and will stay concentrated for a large number of rounds, w.h.p. 
    
    Our protocol has many interesting properties: agents do not need to know other colours and weights in the system, and our protocol requires very little memory per agent. Furthermore, the protocol guarantees fairness meaning that over a long period each agent has each colour roughly a number of times proportional to the weight of the colour. Finally, our protocol also fulfils sustainability meaning that no colour ever vanishes. All of these properties still hold when an adversary adds agents or colours.
 
\end{abstract}


\maketitle

\section{Introduction}
Diversification,
i.e., the action of diversifying something, is omnipresent:
companies balancing the ratio of male to female employees,
farmers sowing a large variety of crops,  investment funds diversifying their portfolios to be more resilient,
bees reallocating tasks to guarantee the survival of the hive, etc.

In order to understand how diversification works on a global scale  through \emph{local choices}, we study the problems from a population protocols perspective: through simple local choices. Consider a population of $n$ agents, where each agent has one of $k$ colours and each colour has an associated weight greater than or equal to $1$. The main question we ask ourselves is: Assuming limited interactions between the agents, does there exist a population protocol that converges to a state where each colour has a support which is, roughly, proportional to its weight?

Before answering such a question, it is relevant to explain why this problem is not trivial. At first we may wonder why we just don't trivially choose a random colour with probability proportional to its weight? There are a couple of reasons for not doing this. First, in order to carry out such simple protocol, each agents needs to know the normalisation constant of the distribution, and thus we need to know all the different colours presented in the system and their respective weights, which requires too much memory and computation. Second, if one colour is removed from the system (e.g. an external agent recolours all red agents blue\footnote{This could arise when a task is fulfilled and no longer necessary.}), other agents will not be notified of this change. The same happens if nature changes the colour of an agent by a completely new one (e.g. an ant notices that the nest temperature is too hot and starts fanning). Therefore, the trivial protocol is not robust to changes in the structure of the population, and therefore it is not good to solve the diversification task. In order to make this protocol more robust, agents would need to constantly broadcast their list of colours, implying a lot of time, memory, local computation, and large messages, likely exceeding the capabilities of simple organisms such as ants. 

In this work, we propose a much more biologically plausible protocol, the \minority protocol, where each agent only requires one extra bit of memory (besides remembering its own colour and its corresponding weight), and agents only need to communicate their own colour and weight. We will show that the \minority protocol achieves \emph{diversity}, and we will also show that our protocol is \emph{fair} and \emph{sustainable}. 
By fairness, we mean that over a long enough period, each agent $u$ has each colour $i\in [k]$ a portion of time which is roughly proportional to the weight $w_i \geq 1, w_i \in \mathbb{R}$ associated to the colour. For example, in the context of task-allocation, this means that all agents perform every task for roughly the same amount of time.

Sustainability refers to the process guaranteeing (with probability $1$) that no colour ever vanishes. Our protocol is also robust to structural changes: even when an adversary adds agents and colours, the protocol quickly returns into a state  of diversity and fairness. Note that none of these properties implies one another.

Our \minority process is not only simple, but also very efficient in converging, as it takes $O(w^2n \log n)$ time-steps to converge, where $w = \sum_{i=1}^k w_i$ is the sum of the $k$ weights of the different colours in the system.  Assuming that the total weight $w$ does not depend on the size of the population, the protocol is even asymptotically optimal, as for example, if a colour is supported by exactly one agent out of $n$, then such colour has to propagate to at least a set of agents of size $\Theta(n)$, which takes at least time $\Omega(n \log n)$ in a population protocol by simple broadcasting arguments.

Interestingly, the questions of diversity, fairness and sustainability  have not been asked in the context of  population protocols, yet.\footnote{There has been some work by Yasumi et al. on finding equi-sized partitions in deterministic population protocol settings. See \cref{sec:related} for more details. }
Instead, research has focused  almost uniquely on understanding the opposite question: how can the population agree on one colour?

\subsection{Related work}\label{sec:related}

 The closest works to ours are \cite{yasumi1,yasumi2,yasumi3} by Yasumi et al.,  where they study population protocols designed to obtain equi-sized partitions. In these works  the authors focused on deterministic  scheduling (worst-case), and investigated several space-complexity problems under several assumptions on, e.g., initial states, power of the (deterministic) adversarial scheduling, etc, where the main objective is to find protocols that use the least number of states. Our approach is rather different as we assume a random schedule, where no adversary is trying to boycott our protocol by slowing down interactions between agents, and thus we focus on finding a fast and simple protocol.
 
Diversity can be seen as an opposite problem of consensus, which is the most studied problem in the field of population protocols. Here $n$ agents start with one of $k$ distinct colours and the goal is to converge quickly to a single colour. The required time is called the \emph{consensus time}.
This finds application in understanding spreading phenomena, for example the
spread of infectious diseases, rumours and opinions in societies, or
viruses/worms in computer network. 
Arguably, the simplest  model is the so-called Voter model, where each agent simply samples a neighbour uniformly at random and adopts its opinion. See \cite{CEOR13,KMS19,OP19} for the state-of-the-art. Other well studied consensus population protocols include
2-Choices and 3-Majority to reach consensus. In the former each agent $u$ samples two other agents in each time-step and updates its colour if both other agents have the same colour. In the latter, $u$ samples two agents and if together with $u$'s own colour there is a majority, then $u$ picks the corresponding colour; otherwise it picks one uniformly at random among the three colours.
See \cite{CER14, BCNPS15,CERRS15, BCNPT16,Cooper2017Fast, MMM18, CNNS18,  Kang19} for references of the state-of-the-art. While 2-Choices and 3-Majority are not exactly a population protocol in the sense that interactions are not pairwise, those protocols have been extended to actual population protocols, with the help of extra memory per agent \cite{DBLP:conf/soda/BecchettiCNPS15}.
The state-of-the-art for population protocols reaching consensus is \cite{BerenbrinkGK20} (with elaborate communication and timing strategies). Other related spreading/consensus processes include the \textsc{Moran} process~\cite{LHN05, L14},
contact processes, and other classic epidemic processes~\cite{BG90, L99, nature16}.

Another protocol closely related to ours is the anti-voter model, where each agent starts with one of two \emph{opposite} colours. In this protocol agents interact in the same way as the Voter model, but agents adopt the \emph{opposite} colour of the sampled neighbour. It is proven that in the long term agents reach an equilibrium and that agents are constantly changing colours \cite{aldousFill14, Rinott97}. However, it is worth noting that this protocol works only for two colours and that agents know the colours, thus it is not straightforward to extend this protocol to more colours, without the use of extra memory by the agents.

Another related type of protocols are averaging processes \cite{MR2908618}, in which upon a pairwise interaction, both agents adopt the average of their opinions. In \cite{wanka}, the authors considered various average processes, in the context of load balancing. Their most related process is the so-called diffusion load-balancing process in the matching model, in which each node starts with some load and in every round the nodes are randomly matched and average their load. The authors obtained bounds on the convergence time and the load discrepancy over time. Recently, authors of \cite{MMP19} considered the noisy-averaging population protocol. Here, each agent has an initial value and the goal is to agree on the average. The catch is that the communication is noisy and the value communicated by the agents can be altered. The authors studied the behaviour of the simple averaging protocol in which two agents are chosen uniformly at random and both set their values to the average of the received value and their own value.

\medskip

Population protocols have also been used in the context of community detection. Most of the research here has focused on the stochastic block model where each agent belongs to one of two communities. The agents from the same community have an edge with probability $p$ and agents from different communities have an edge with probability $q$, where $q<p$. The goal is to recover the hidden partition.
See
\cite{MMM18,FriendorFoe,CNS19} for more details.

There has also been a large body of research on task allocation. Note that our problem is a special case of task allocation where each task has a demand proportional to $w_i$. See 
\cite{alex, ants,radeva2017costs} for ant-inspired research on task allocation. The main difference from our work is that this line of research assumes that the ants receive  feedback from the tasks indicating whether there are too many or too few ants working on the task. In our setting, there is no such explicit feedback; agents receive implicit feedback through sampling the current colour of other agents.

\subsection{Model and protocol}
We study the following population protocol.
Initially, there are $n$ agents, each of which has one colour out of $k$ colours; and each colour $i \in [k]$ has an associated weight $w_i \geq 1$.  The system evolves as follows: at every time-step an agent $u$ is scheduled u.a.r. (uniformly at random). The scheduled agent $u$ then samples another agent $v$ u.a.r, and \textit{observes} the colour of $v$, as well as its weight, and then processes such information, leading to a possible change of colour and the associated weight. We denote by $c_u(t)$ the colour of agent $u$ after the $t$-th time-step, and $c_u(0)$ its initial colour. Additionally, we denote by $C_i(t)$ the number of agents with colour $i$ at time-step $t$, i.e., $C_i(t) = |v\in [n]: c_v(t) = i|$.

We now introduce the properties we will study in this paper: diversity, fairness, and sustainability.

\begin{definition}\label{def:props}\text{}

\begin{enumerate}
    \item {\bf{Diversity}}.  A protocol is diverse if there exists a $t_0\geq 0$ and $T = \Omega(n^{\alpha})$ such that w.h.p.\footnote{We use w.h.p. to denote ``with high probability" meaning w.p. at least $1-1/n$.} 
\begin{align}\label{eqn:defiDiver}
\left|\frac{C_i(t)}{n}- \frac{w_i}{w}\right | = \tilde O\left(\frac{1}{\sqrt{n}}\right)\,,
\end{align}    
for any $t \in \{t_0,\ldots, T\}$  and $i\in [k]$. In words, the population stabilises in a subset of configurations such that each colour  $i$ appears in the population proportionally to its weight $w_i$ for $\Omega(n^{\alpha})$ time-steps. Note that if $w_i=1$ for all colour $i$, then achieving diversity is equivalent to finding a uniform partition of the vertices.

    \item {\bf  Fairness}.   A protocol is fair  if there exists $T = \Omega(n^{\beta})$ such that for any $T'>T$, it holds that  w.h.p  \[  \frac{\left|\{  t' \in[0, T'] \colon c_{u}(t')=i \}\right|}{T'}
    =(1\pm o(1)) \frac{w_i}{w} \,,\] 
    for any $i\in [k]$ and any agent $u$. In  words, each agent has, roughly, colour $i$ for a proportion $w_i/w$ of the times.
    \item {\bf Sustainability}.
    A protocol is sustainable if for all $t>0$ and any colour $i$, there exists at least one agent with colour $i$ at time $t$. In other words, no colour ever vanishes.
\end{enumerate}
If a protocol is diverse, fair, and sustainable, it is called \textbf{good}.
\end{definition}

\paragraph{\minority Protocol}We proceed to define the \minority protocol, which satisfies all the above properties. In the \minority protocol all agents have an extra bit of memory, which can be observed by other agents when they interact. This bit represents the degree of confidence they have in their current colour. Agents whose bit is 1 will not change their current colour, whereas agents whose bit is 0 are open to change their current colour. Pictorially, we call colours with bit 0 \emph{light}, and those with bit 1 \emph{dark}, both of which are also referred to as \emph{shade}. Therefore, a dark colour needs to become light before changing into a completely different one. Let $b_u(t)$ represents the extra bit of memory of agent $u$ after the $t$-th iteration. We assume that $b_u(0) = 1$ for all agents $u$. The \minority protocol is defined as follows: suppose that in the $(t+1)$-th time-step agent $u$ is scheduled and it samples another agent $v$ u.a.r. If $u$ has a light colour, and $v$ has a dark colour, then $u$ adopts the colour of $v$ and its shade. If both $u$ and $v$ have the same dark colour $i$, then $u$ changes its colour to light (i.e. it changes the value of the extra bit to $0$) with probability $1/w_i$\,, where $w_i \geq 1$ is the associated weight of colour $i$\,. Formally, the changes occur according to the following rule:
\begin{align}\label{eqn:definitionDiver}
(c_{u}(t+1),b_u(t+1)) = \begin{cases} 
      (c_v(t),1) & \text{if } b_u(t) = 0 \text{ and } b_v(t) =1\,, \\
      (c_u(t),0) ~ w.p.~\frac{1}{w_{c_u(t)}} & \text{if } b_u(t)=b_v(t) = 1\,, \text{ and } c_u(t) = c_v(t)\,, \\
      (c_u(t),b_u(t)) &  \text{otherwise}. 
   \end{cases}
\end{align}

Note that when all weights are equal to $1$, our protocol gives a deterministic protocol for the uniform partition problem.

The intuition of the protocol, from the agent point of view, is that if agent $u$ observes another agent $v$ with the same colour, then this suggests that this colour is over-represented. Of course, using only one observation to decide if one colour is over-represented is very crude.
Nonetheless, in expectation this approach works: the colour distribution approaches the target distribution (e.g., uniform distribution if all weights are $1$).
Now, since we have $n$ agents performing the same protocol over long periods, emergent behaviour appears:
the colour distribution will be concentrated around its expected value, the target distribution.

To ensure that the protocol indeed converges towards the target distribution, two rules are important. 
The first one is that only light shaded colours can change their colour (first line of \cref{eqn:definitionDiver}), the second rule is that dark shaded colours change their shade to light if they observe another same colour with dark shade with probability equals to the inverse of the weight of the colour, otherwise they keep the shade (second line of \cref{eqn:definitionDiver}). 
This causes heavy weighted colours to change less often than light weighted ones, and thus we expect to see more of the former than the latter.

A rough idea of why this protocol works is that for each colour $i$, the 'rate' at which colours $i$ decrease by 1 is $\frac{1}{n^2}C_i(t)^2/(w_i)$, whereas the 'rate' $i$ increases by $1$ is $\frac{1}{n^2} C_i(t) \sum_{j=1}^{k}C_j(t)/w_j$, hence equilibrium is achieved when all values $C_i(t)/w_i$ are roughly the same, that is $C_i \approx w_i n/w$, leading to diversity. 
Fairness is obtained by noting that an agent changes its colour to colour $j$ with probability proportional to $C_j(t)$, and since the protocol eventually reaches a configuration with $C_j(t) \approx w_j n/w$, we have that for large $t$, an agent has colour $j$ with probability, approximately, a $w_j/w$ fraction of the time. Sustainability simply follows from the fact that an agent with a dark shade of a colour changes its colour only when the agent encountered another agent with the same dark shaded colour. Since only one agent changes its colour at a given time step, no dark shaded colour can ever vanishes .Moreover, Sustainability is preserved if new agents are added to the system, or new colours are added, as long as the new colours are initially dark, and they do not replace the last dark version of another colour.

In the next section we will formalise these ideas, and we will show that the \minority protocol satisfies the Diversity and Fairness properties. 

\paragraph{Derandomisation} Before finishing this section, we describe a derandomised version of our protocol that avoids the sampling procedure in the second line of \cref{eqn:definitionDiver}. Here we shall assume that all weights are non-negative integers. Then, the derandomised \minority protocol is described as following: Instead of having light and dark shades, we will have $1+w_i$ different shades of grey  which are enumerated from $0$ (light) to $w_i$ (dark). Whenever agent $u$ that has colour $i$ and shade greater than 0 is scheduled and it chooses another agent with the same colour and shade greater than $0$, then $u$ reduces the shade of its colour by $1$. If an agent $u$ has colour $i$ and shade $0$ chooses an agent $v$ with shade greater than 0, then $u$ adopts the colour of $v$, say $j$, and sets its shade to $w_j
$. In all other interactions, the agents do nothing. Note this algorithm requires $\lceil \log_2(1+w_i)\rceil$ extra bits of memory when it adopts colour $i$\,.
\subsection{Results}
For simplicity, in this work we assume that $k$ and $w$ are constants, however, we state most of the intermediate results in terms of $k$ and $w$, but we do not attempt to optimise the terms involving $k$ or $w$.

\begin{theorem}\label{thm:main}
For constant $k$ and $w$, the \minority protocol is good, i.e., it achieves diversity, fairness and sustainability.
\end{theorem}

To prove the previous theorem we just need to prove that the protocol is diverse and fair as we already argued that the protocol achieves sustainability. For diversity, we rely heavily on the following theorem, which directly gives the desired result.
\begin{theorem}\label{thm:easy}
There exists $T = O(w^2n\log n)$, such that w.h.p. it holds that for every colour $i$ we have
\[ \sum_{i=1}^k \sum_{j=1}^k\left|\frac{C_i(t)}{w_i}-\frac{C_j(t)}{w_j}\right|^2 = O(wn\log n)\,,\]
for all $t$ in the interval $[T, n^8]$.
\end{theorem}
Dividing by $2k^2$ both sides of the previous equation, expanding the squared, and rearranging some terms yields
\begin{align}\label{eqn:varversion}
    \frac{1}{k}\sum_{i=1}^k \left(\frac{C_i(t)}{w_i}- x\right)^2 = O(wn(\log n)/k^2)\,,
\end{align}

where $x = (1/k)\sum_{i=1}^k C_i(t)/w_i$. From here we will deduce that
\begin{align}\label{eqn:after2}
    \frac{C_i(t)}{w_i} = \frac{n}{w}+O(k\sqrt{wn\log n})
\end{align}
 holds after $O(n\log n)$ time-steps for at least  $\Omega(n^8)$ time-steps, 
with high probability. To derive \cref{eqn:after2} we just need to verify that $x = n/w + O(k\sqrt {w n \log n})$. For that notice that \cref{eqn:varversion} yields
\begin{align}
    C_i(t) = w_ix + w_i O(\sqrt{wk n\log n})
\end{align}
hence, by using that $\sum_{i=1}^k C_i(t) = n$, we get
\begin{align}
    n = wx + w O(\sqrt{wk n\log n})
\end{align}
from which we deduce that $x = \frac{n}{w}+ O(\sqrt{wk n\log n})$, and thus \cref{eqn:after2} holds. 

The previous result implies that the \minority protocol is diverse. While we proved that $\alpha$ in the definition of diversity can be chosen as $\alpha = 8$, a simple inspection of our proof shows that $\alpha$ can be chosen as an arbitrarily large constant. The proof of  \cref{thm:easy} can be found in \cref{sec:mainproof}\,.

The proof of fairness is provided in \cref{sec:props}\,, which is built on the results proved in  \cref{sec:mainproof}\,. The main idea is that after $O(n\log n)$ time-steps the whole system stabilises for at least $\Omega(n^8)$ time-steps in configurations where colours (including light and dark versions) are almost perfectly distributed among the agents (according to the appropriate weight values). Therefore, the plan is that, instead of following the colours of agents in the \minority protocol, we follow agents in a new system where the (dark and light) colours are perfectly distributed among the agents.

\subsection{Summary of contributions and main techniques}
The main contribution of this paper is to propose and demonstrate that the \minority protocol achieves diversification, fairness, and sustainability. The analysis of the protocol is mainly probabilistic, and it is divided into three phases. In the first phase the main technique is to couple certain statistics of our protocol with biased random walks on the integers in order to achieve rapid convergence (despite the fact that this is the slowest part of the process as the initial configurations of colours are arbitrary). In the second phase we introduce two potential functions that decrease over time, and introduce a general concentration inequality to analyse those potentials. Finally, in the third and final phase we couple the trajectory of the states of agents, that is $(c_u(t), b_u(t))_{t\geq 0}$ with a Markov chain $P$ that represents the system in ``perfect equilibrium", and show that both processes roughly hit every colour  the same amount of times. To prove our coupling is correct we make use of Chernoff's Bounds for Markov Chains.

\section{Analysis of the \minority protocol}\label{sec:mainproof}


We begin by setting some notations that will be used throughout the analysis of our protocol.
For simplicity, we assume that the colours are enumerated from $1$ to $k$ with corresponding weights $w_1, \ldots, w_k$\,, and we set $w = \sum_{i=1}^k w_i$\,. Recall that $c_u(t)$ denotes the colour of agent $u$ after $t$-th iterations, and $b_u(t)$ its corresponding shade indicator, 1 for dark shade, 0 for light shade. We denote by $A_i(t)$ and $a_i(t)$ the number of agents having colour $i$ at time $t$ with dark shade (bit value $1$) and light shade (bit value $0$), respectively, i.e. for $t > 0$\,,
\begin{align*}
a_i(t) &= |\{u: c_u(t) = i, b_u(t) = 0\}|,
\end{align*}
and
\begin{align*}
A_i(t)&= |\{u: c_u(t) = i, b_u(t) = 1\}|\,.
\end{align*}

Define $A(t) = \sum_{i=1}^k A_i(t)$ and $a(t) = \sum_{i=1}^k a_i(t)$, and denote by $\xi(t)$ the process containing all the information at time $t$, that is,
\begin{align*}
\xi(t)= (A_1(t),\ldots, A_k(t),a_1(t),\ldots, a_k(t))\,.
\end{align*}
Let $\Omega = \{(A_1,\ldots,A_k, a_1,\ldots, a_k): A_i\geq 1, a_i\geq 0, \sum_{i=1}^k(A_i+a_i) = n\}$ be the space state where the process $\xi(t)$ takes values. Clearly the initial state is such that $A_i(0) \geq 1$ for all colours $i$. 

As is usually done in population protocol, our analysis is divided in several phases. In each of these phases the configuration of the system (i.e., the colours of the agents) will be attaining properties that are maintained for long periods of time (much larger than the duration of those phases), and each successive phase will make use of the previously achieved properties. The main idea revolves around showing that perfect equilibrium is attained at the values
\begin{align}\label{eqn:goalaAA}
a_i(t)/n =\frac{w_i/w}{1+w} \;\text{ and }\;A_i(t)/n = \frac{w_i}{1+w}\,,
\end{align}
for all $i \in \{1,\ldots, k\}$. 
The analysis of our protocol requires three phases which we proceed to describe.

\textbf{Phase 1} starts at time 0 and ends at time $\tau_1 = O(w^2 n \log n)$. In this phase we will show that w.h.p. the equality in \cref{eqn:goalaAA} can be attained up to multiplicative constants (see \cref{thm:phase1thm}). This property is shown to hold for at least $\Omega(n^{10})$ time-steps, enough to be carried to the next phases of the analysis. The analysis of this phase is in \cref{sec:phase1}.

\textbf{Phase 2} starts exactly at the end of Phase 1, and lasts for $\tau_2 = O(wn \log n)$ time-steps. This phase is divided in two consecutive subphases. The first subphase lasts for $\tau_{2,1} = O(wn \log n)$ time steps, while the second subphase lasts for $\tau_{2,2} = O(wn \log n)$ time steps. Clearly $\tau_2 = \tau_{2,1}+\tau_{2,2}$\,.

In this phase we will analyse two potential functions $\phi$ and $\psi$\,, given by
\begin{align*}
\phi(t) = \sum_{i=1}^k\sum_{j=1}^k \left(\frac{A_i(t)}{w_i}-\frac{A_j(t)}{w_j} \right)^2,
 \text{ and } 
\psi(t)= \sum_{i=1}^k \sum_{j=1}^k \left( \frac{a_i(t)}{w_i}-\frac{a_j(t)}{w_j}\right)^2.
\end{align*}

Then, in the first subphase we show that the potential $\phi(t)$ decreases to $O(wn\log n)$ w.h.p. (\cref{lemma:firstPotentialSmall}), and in the second subphase, $\psi(t)$ is shown to decrease to $O(w n\log n)$ w.h.p. (\cref{lemma:secondPotentialSmall}). These properties hold for at least $n^8$ time steps (see \cref{thm:easyprima}). The analysis of Phase 2 is in \cref{sec:phase2}.

The analysis of the first two phases is enough to prove diversity, however, to prove fairness we need an extra phase.

\textbf{Phase 3} starts exactly after the end of Phase 2, and lasts for $\tau_3=O(wn \log n)$ time-steps. In this phase we show that the system is very close to perfect equilibrium, indeed, \cref{eqn:goalaAA} holds up to a small additive error, and such property holds up to time $n^8$ whp (see \cref{thm:approximation}). To get such approximation, we will show that the potential function $\sigma^2$ given by
\begin{align*}
    \sigma^2(t) = (A(t)/w - a(t))^2
\end{align*}
decreases to $O(w n \log n)$, which combined with the fact that the potential functions $\phi(t)$ and $\psi(t)$ are already of the same size, gives us enough information to obtain the desired approximation. Phase 3 is analysed in \cref{sec:props}.

For quick reference, \cref{fig:figuree} below shows a short summary of the three phases.

\begin{figure}[H]
    \centering
    \begin{tikzpicture}[x=0.75pt,y=0.75pt,yscale=-1,xscale=1]

\draw    (7,132.67) -- (107.67,132.67) ;
\draw [shift={(110.67,132.67)}, rotate = 180] [fill={rgb, 255:red, 0; green, 0; blue, 0 }  ][line width=0.08]  [draw opacity=0] (8.93,-4.29) -- (0,0) -- (8.93,4.29) -- cycle    ;
\draw [shift={(4,132.67)}, rotate = 0] [fill={rgb, 255:red, 0; green, 0; blue, 0 }  ][line width=0.08]  [draw opacity=0] (8.93,-4.29) -- (0,0) -- (8.93,4.29) -- cycle    ;
\draw [line width=2.25]  [dash pattern={on 6.75pt off 4.5pt}]  (108,8) -- (110,125) ;
\draw [line width=2.25]  [dash pattern={on 6.75pt off 4.5pt}]  (287,6.5) -- (289,123.5) ;
\draw [line width=2.25]  [dash pattern={on 6.75pt off 4.5pt}]  (398,7) -- (400,124) ;
\draw    (291.67,131.67) -- (397.67,131.67) ;
\draw [shift={(400.67,131.67)}, rotate = 180] [fill={rgb, 255:red, 0; green, 0; blue, 0 }  ][line width=0.08]  [draw opacity=0] (8.93,-4.29) -- (0,0) -- (8.93,4.29) -- cycle    ;
\draw [shift={(288.67,131.67)}, rotate = 0] [fill={rgb, 255:red, 0; green, 0; blue, 0 }  ][line width=0.08]  [draw opacity=0] (8.93,-4.29) -- (0,0) -- (8.93,4.29) -- cycle    ;
\draw    (113.67,132.66) -- (196.6,132.41) ;
\draw [shift={(199.6,132.4)}, rotate = 539.8299999999999] [fill={rgb, 255:red, 0; green, 0; blue, 0 }  ][line width=0.08]  [draw opacity=0] (8.93,-4.29) -- (0,0) -- (8.93,4.29) -- cycle    ;
\draw [shift={(110.67,132.67)}, rotate = 359.83] [fill={rgb, 255:red, 0; green, 0; blue, 0 }  ][line width=0.08]  [draw opacity=0] (8.93,-4.29) -- (0,0) -- (8.93,4.29) -- cycle    ;
\draw    (202.6,132.38) -- (285.67,131.69) ;
\draw [shift={(288.67,131.67)}, rotate = 539.53] [fill={rgb, 255:red, 0; green, 0; blue, 0 }  ][line width=0.08]  [draw opacity=0] (8.93,-4.29) -- (0,0) -- (8.93,4.29) -- cycle    ;
\draw [shift={(199.6,132.4)}, rotate = 359.53] [fill={rgb, 255:red, 0; green, 0; blue, 0 }  ][line width=0.08]  [draw opacity=0] (8.93,-4.29) -- (0,0) -- (8.93,4.29) -- cycle    ;
\draw [line width=1.5]  [dash pattern={on 1.69pt off 2.76pt}]  (198,45) -- (199,122) ;
\draw    (111.67,152.65) -- (284.67,151.68) ;
\draw [shift={(287.67,151.67)}, rotate = 539.6800000000001] [fill={rgb, 255:red, 0; green, 0; blue, 0 }  ][line width=0.08]  [draw opacity=0] (8.93,-4.29) -- (0,0) -- (8.93,4.29) -- cycle    ;
\draw [shift={(108.67,152.67)}, rotate = 359.68] [fill={rgb, 255:red, 0; green, 0; blue, 0 }  ][line width=0.08]  [draw opacity=0] (8.93,-4.29) -- (0,0) -- (8.93,4.29) -- cycle    ;
\draw    (0,121.67) -- (480.67,121.67) ;
\draw [color={rgb, 255:red, 208; green, 2; blue, 27 }  ,draw opacity=1 ]   (182.45,40.92) -- (156,57.33) ;
\draw [shift={(185,39.33)}, rotate = 148.17] [fill={rgb, 255:red, 208; green, 2; blue, 27 }  ,fill opacity=1 ][line width=0.08]  [draw opacity=0] (8.93,-4.29) -- (0,0) -- (8.93,4.29) -- cycle    ;
\draw [color={rgb, 255:red, 208; green, 2; blue, 27 }  ,draw opacity=1 ]   (213.53,40.94) -- (241,58.33) ;
\draw [shift={(211,39.33)}, rotate = 32.35] [fill={rgb, 255:red, 208; green, 2; blue, 27 }  ,fill opacity=1 ][line width=0.08]  [draw opacity=0] (8.93,-4.29) -- (0,0) -- (8.93,4.29) -- cycle    ;

\draw (469,121) node [anchor=north west][inner sep=0.75pt]   [align=left] {time};
\draw (21,3) node [anchor=north west][inner sep=0.75pt]  [font=\large] [align=left] {\textbf{Phase 1}};
\draw (45,134) node [anchor=north west][inner sep=0.75pt]   [align=left] {$ \tau_{1}$};
\draw (10,55) node [anchor=north west][inner sep=0.75pt]  [font=\scriptsize] [align=left] {$O( 1)$ multiplicative-\\error approximation\\of Eq. (3)};
\draw (341,134) node [anchor=north west][inner sep=0.75pt]   [align=left] {$ \tau_{3}$};
\draw (140,134) node [anchor=north west][inner sep=0.75pt]   [align=left] {$\tau _{2,1}$};
\draw (222.74,134) node [anchor=north west][inner sep=0.75pt]   [align=left] {$\tau _{2,2}$};
\draw (138,78) node [anchor=north west][inner sep=0.75pt]  [font=\scriptsize] [align=left] {$ \phi(t)$};
\draw (121,92) node [anchor=north west][inner sep=0.75pt]  [font=\scriptsize] [align=left] {decreases to\\O($ wn\ \log n)$};
\draw (228,78) node [anchor=north west][inner sep=0.75pt]  [font=\scriptsize] [align=left] {$  \psi(t)$};
\draw (211,93) node [anchor=north west][inner sep=0.75pt]  [font=\scriptsize] [align=left] {decreases to\\O($ wn\ \log n)$};
\draw (20,22) node [anchor=north west][inner sep=0.75pt]  [color={rgb, 255:red, 0; green, 117; blue, 255 }  ,opacity=1 ] [align=left] {\textbf{(Thm~\ref{thm:phase1thm}.)}};
\draw (166,3) node [anchor=north west][inner sep=0.75pt]  [font=\large] [align=left] {\textbf{Phase 2}};
\draw (165,22) node [anchor=north west][inner sep=0.75pt]  [color={rgb, 255:red, 0; green, 117; blue, 255 }  ,opacity=1 ] [align=left] {\textbf{(Thm~\ref{thm:easyprima}.)}};
\draw (122,60) node [anchor=north west][inner sep=0.75pt]  [font=\small,color={rgb, 255:red, 208; green, 2; blue, 27 }  ,opacity=1 ] [align=left] {\textbf{Lemma ~\ref{lemma:firstPotentialSmall}}};
\draw (212,60) node [anchor=north west][inner sep=0.75pt]  [font=\small,color={rgb, 255:red, 208; green, 2; blue, 27 }  ,opacity=1 ] [align=left] {\textbf{Lemma~\ref{lemma:secondPotentialSmall}}};
\draw (305,4) node [anchor=north west][inner sep=0.75pt]  [font=\large] [align=left] {\textbf{Phase 3}};
\draw (302,23) node [anchor=north west][inner sep=0.75pt]  [color={rgb, 255:red, 0; green, 117; blue, 255 }  ,opacity=1 ] [align=left] {\textbf{(Thm~\ref{thm:approximation}.)}};
\draw (302,57) node [anchor=north west][inner sep=0.75pt]  [font=\scriptsize] [align=left] {$ o( 1)$ additive-error\\approximation\\of Eq. (3)};
\draw (424,39) node [anchor=north west][inner sep=0.75pt]  [font=\scriptsize] [align=left] {W.h.p\\All previous\\properties\\hold up to \\time $ n^{8}$};
\draw (190,154) node [anchor=north west][inner sep=0.75pt]   [align=left] {$ \tau _{2}$};

\end{tikzpicture}
    \caption{Summary of the three phases of the analysis.}
    \label{fig:figuree}
\end{figure}
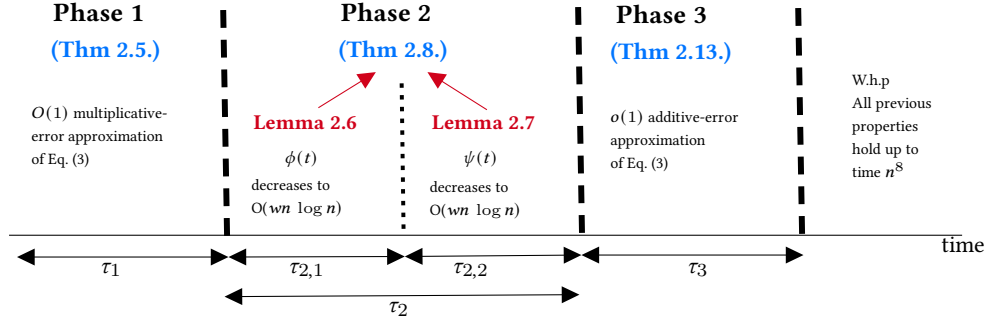

\subsection{Phase 1: The rise of the minorities }\label{sec:phase1}

In the first phase of the analysis, we handle colours that are largely over-represented or largely under-represented, and we will show that the system stabilises in a configuration such that for any $i \in[k] \,,$
\begin{align}\label{eqn:goalaA}
a_i(t)/n \approx\frac{w_i/w}{1+w} \;\text{ and }\;A_i(t)/n \approx \frac{w_i}{1+w}\,,
\end{align}
where the approximation holds up to a multiplicative constant. For this, we will demonstrate that the process, little by little, starts improving its configuration, in the sense that we are closer to achieve \cref{eqn:goalaA} over time. At the beginning the process is rather slow at making progress as some colours may have only one agent supporting them, but as long as the configuration improves, the speed at which the process improves the configuration increases as well. To discuss this in detail, we first show that $a(t)=\sum_{i=1}^k a_i(t)$ quickly increases to roughly $ (1-\eps)n /(1+w)$ (recall that $a(0) = 0$). We achieve this by coupling the process $a(t)$ with a biased random walk. To see why $a(t)$ has bias towards increasing its value, note that at the beginning  all agents have dark colours, and therefore interactions between agents of the same colour are likely to be between agents of the same shade, thus enlarging the chance of increasing the value of $a(t)$. After obtaining a healthy proportion of light-coloured agents (which are available to change their current colour to any other colour), we can show that under-represented colours start growing by another biased random walk argument. At first, we will show that under-represented dark shaded colours, say $A_i$, have bias toward increasing their representation: At the beginning, it is unlikely for an agent to sample an under-represented colour $i$ and therefore it is unlikely that an agent switches to $i$, however, it is even less likely that an agent of colour $i$ switches to a different colour, so at the beginning, $A_i(t)$ grows very slowly. Then, as the popularity of colour $i$ grows, it becomes more likely to sample $i$, thus $A_i(t)$ will start to grow increasingly faster, until $A_i(t)/n$ is close to $w_i/(1+w)$\,, at which point the bias of the process start being less relevant to make significant progress. The same can be applied to each under-represented light shaded colour, say $a_i(t)$, and also to over-represented colours with dark and light shades respectively. By the end of Phase 1 the bias presented in the system is very small, and thus coupling with biased random walks are not enough to prove improvements in the approximation of \cref{eqn:goalaA}.

For the formal analysis it is convenient to define the following regions in terms of a constant parameter $\varepsilon\in (0,1/4)$:
\begin{align*}
R_1 &= \left\{\xi \in \Omega: \frac{a}{n} \geq (1-\varepsilon)\frac{1}{w+1}\right\}
\nonumber\\
S_1 &=\left\{\xi \in \Omega: \frac{a}{n} \geq (1-2\varepsilon)\frac{1}{w+1}\right\}
\nonumber\\
R_2 &= \left\{\xi \in \Omega: \forall i\in[k], ~\frac{A_i}{n} \geq (1-3\varepsilon)\frac{w_i}{1+w}\right\} \cap S_1
\nonumber\\
S_2 &= \left\{\xi \in \Omega: \forall i\in[k], ~\frac{A_i}{n} \geq (1-4\varepsilon)\frac{w_i}{1+w}\right\} \cap S_1
\nonumber\\
S_3 &= \left\{\xi \in \Omega: \forall i\in[k], ~ \frac{A_i}{n} \leq \left(1+4\varepsilon w\right)\frac{w_i}{1+w}\right\} \cap S_2
\nonumber\\
S_4 &= \left\{\xi \in \Omega: \frac{a}{n} \leq \left(1+4\varepsilon w\right)\frac{1}{1+w}\right\} \cap S_3 \,.
\end{align*}

Note that $R_j \subseteq S_j$\,. Define
$T_0 = 0$ and $T_{i} = \min\{t\geq T_{i-1}: \xi(t) \in R_i\}$\,, and define $T_i' = \min\{t\geq T_i: \xi(t) \in \Omega \setminus S_i\}$\,.

\begin{lemma}\label{lemma:part1}
Let $c>0$ be a sufficiently small constant.
For any $\xi(0) \in \Omega$\,, we have that $T_{1} = O(nw/\eps)$ with probability at least $1-\exp(c \eps n)$\,. Moreover, if $\xi(0) \in R_1$ then $T_1'>\exp(cn\eps^2/w)$ with probability at least $1-\exp(-c n\eps^2/w)$\,.
\end{lemma}

\begin{lemma}\label{lemma:part2}
Let $c>0$ be a sufficiently small constant.
For any $\xi(0) \in \Omega$\,, we have $T_2 = O\left( w n\log n / \eps\right)$ with probability at least $1-\exp\left(- c n\eps^2 / w \right)$\,. Additionally, if $\xi(0) \in R_2$ then $T'_{2} > \exp(cn\eps^2/w)$ with probability at least $1-\exp\left(-c n \eps^2 /w \right)$\,.
\end{lemma}

\begin{lemma}\label{lemma:part3}
Assume that $\xi(t) \in S_2$\,. Then, $\xi(t) \in S_3$\,.
\end{lemma}
\begin{proof}
Using the lower bounds, we have,
$A_i \leq n - a - \sum_{j \neq i} A_j \leq n- (1-2\eps)n/(w+1) - (1-4\eps)n(w-w_i)/(1+w)\leq n \frac{w_i}{w+1} + n 4\eps\frac{w}{w+1}$\,.

\end{proof}

\begin{lemma}
Assume that $\xi(t) \in S_3$\,. Then, $\xi(t) \in S_4$\,.
\end{lemma}
\begin{proof}
Using the lower bounds, we have,
$ a\leq n -\sum_{i=1}^k A_i \leq (1-4\eps )n w /(1+w) \leq \frac{n}{w+1} + 4\eps n w /(1+w)$\,.
\end{proof}

\bigskip
For any $\delta>0$ define the set of configurations $\mathcal E = \mathcal E(\delta) \subseteq \Omega$ by
\begin{align}\label{eqn:defi_E}
\mathcal E(\delta) = \left\{\xi \in \Omega: \frac{A_i}{w_i} \in \left[\frac{(1-\delta)n}{1+w},\frac{(1+\delta)n}{1+w} \right] \forall i \in [k], \text{ and } a \in \left[\frac{(1-\delta)n}{1+w},\frac{(1+\delta)n}{1+w} \right] \right\}\,.
\end{align}

By applying the previous lemmas in order with $\varepsilon = \frac{\delta}{4w}$ we obtain the following result.

\begin{theorem}\label{thm:phase1thm} Let $\delta>0$ be fixed, then there exists $\tau_1=O(w^2n\log n)$ such that 
\begin{align*}
\Prob\left(\bigcap_{t=\tau_1}^{\tau_1+n^{10}} \left\{\xi(t) \in \mathcal E \right\}\right)\geq 1-\exp\left(-\Omega(n/w^3)\right)\,.
\end{align*}
\end{theorem}

Note that $\tau_1$ in the previous theorem implicitly depends on the choice of $\delta$\,. For our purposes, we just fix $\delta$ as a small enough constant (for our analysis $\delta = 0.0001$ is more than enough). The second phase of the process starts exactly at time $\tau_1$ of \cref{thm:phase1thm}\,, where we know that w.h.p. $A_i(t)/w_i$ and $a(t)$ are equal up to multiplicative constants, and this is valid for at least $n^{10}$ time-steps.

The proof of the previous theorem and the respective lemmas can be found in \cref{sec:proofsphase1}.

\subsection{Phase 2: Reaching Equilibrium - Proving Diversity}\label{sec:phase2}

The second phase starts when the proportions $a_i/w_i$ are all equal up to a multiplicative constant, as well as the proportions $A_i/w_i$. It is convenient to restart time back to $0$ for the analysis of phase 2 of the process. Therefore we set the starting configuration of phase 2 as $\xi(0)$\,. We know that $\xi(0) \in \mathcal E$ and it has the property $\xi(t) \in \mathcal E$ for at least $n^{10}$ time-steps.

As it was mentioned at the beginning of this section, in the second phase we will analyse two potential functions $\phi$ and $\psi$\,, given by
\begin{align}\label{eqn:potential1Defi}
\phi(t) = \sum_{i=1}^k\sum_{j=1}^k \left(\frac{A_i(t)}{w_i}-\frac{A_j(t)}{w_j} \right)^2 \,,
\end{align}
and
\begin{align}\label{eqn:potential2Defi}
\psi(t)&= \sum_{i=1}^k\sum_{j=1}^k \left( \frac{a_i(t)}{w_i}-\frac{a_j(t)}{w_j}\right)^2\,.
\end{align}
Also, recall that Phase 2 is divided into two consecutive subphases, Subphase 2.1 and Subphase 2.2, lasting $\tau_{2,1}$ and $\tau_{2,2}$ time-steps, respectively. In the first one we prove that $\phi(t)$ reduces its value to $O(wn\log n)$, and in the second one we show the same for $\psi(t)$.

It is worth mentioning that both potential functions depend on each other implicitly. Potential function $\phi$ only depends on terms containing $A_i(t)$'s, however, there is an implicit dependency on light shaded colours, as only light shaded agents can transform into dark shaded ones. The same holds for the potential function $\psi$.  Fortunately, the properties proved in Phase 1 of our analysis (that hold for very long periods of time) allow us to control the implicit dependency between the potentials. In particular, we will verify that both $\phi$ and $\psi$ are approximately super-martingales, and both present a drift towards reducing their value, and indeed, we will show that their values halve every $O(wn)$ time-steps. In order to obtain strong w.h.p. bounds for $\phi$ and $\psi$, we introduce a Chung-Lu-type concentration to bound a general class of processes containing $\phi$ and $\psi$ (see \cref{lemma:processConcentration}), which may be of independent interest. Using our new bounds, we are able to show that the potentials quickly reach size $O(wn \log n)$ (\cref{lemma:firstPotentialSmall} and \cref{lemma:secondPotentialSmall}). 

\bigskip
Recall that we restart time at the beginning of Phase 2. For the analysis of Subphase 2.1 we define the event $B_t$ by
\begin{align}
B_t= \{\xi(s) \in \mathcal E\,, \text{ for all } s \in \{0,1,\ldots, t-1\}\}\,,
\end{align}
then we have the following result.

\begin{lemma}\label{lemma:firstPotentialSmall}
Given $r>10$\,, there exists $C>0$ and $\tau_{2,1} = O(w n \log n)$ such that with probability at least $n^{10-r}$\,,
\begin{align}
   \ind_{B_t} \sum_{i=1}^k\sum_{j=1}^k \left(\frac{A_i(t)}{w_i}-\frac{A_j(t)}{w_j}\right)^2 \leq Cwn\log n \,,
\end{align}
for all $t \in \{\tau_{2,1},\tau_{2,1}+1,\ldots, \tau_{2,1}+n^9\}$\,.
\end{lemma}

Note that in the above lemma we prove the result for $\ind_{B_t}\phi(t)$ instead of $\phi(t)$\,. This is just for the sake of the formal analysis, as we know from Phase 1 that the events $B_t$ hold true w.h.p..

Similarly, we can prove that the second potential $\psi$ also decreases very fast to a value of $O(wn)$\,. However, we can only prove this after Subphase 2.1 is over as our analysis uses the fact that $\phi(t)$ is already small. Again, it is convenient to restart time at $\tau_{2,1}$ (i.e. we start Subphase 2.1 at time 0).

Define the set of configurations $\mathcal E' \subseteq \Omega$ as
\begin{align}\label{eqn:defiE'}
\mathcal E' = \left\{\xi \in \Omega: \sum_{i=1}^k \sum_{j=1}^k \left(\frac{A_i}{w_i}-\frac{A_j}{w_j}\right)^2 \leq Cwn\right\} \cap \mathcal E \;\,,
\end{align}
where $\mathcal E$ is defined in~\cref{eqn:defi_E}\,, and define the event $B'_t = \{\xi(s) \in \mathcal E' \,, \forall s\in \{0,\ldots, t-1\}\}$,
which we know it holds true after Phase 2.1 for at least $n^9$ time-steps by~\cref{lemma:firstPotentialSmall}.

\begin{lemma}\label{lemma:secondPotentialSmall}
Given $r>9$\,, there exists $C'>0$ and $\tau_{2,2} = O(wn \log n)$ such that with probability at least $n^{9-r}$\,,
\begin{align}
    \ind_{B'_t}\sum_{i=1}^k\sum_{j=1}^k \left(\frac{a_i(t)}{w_i}-\frac{a_j(t)}{w_j}\right)^2 \leq C' wn\,,
\end{align} 
for all $t \in \{\tau_{2,2},\ldots, \tau_{2,2}+n^8\}$\,.
\end{lemma}

Again, note that we prove the result for the process $\ind_{B_t'}\psi(t)$ instead of $\psi(t)$\,, but, as before, this is just for having simpler analysis as in practice the events $B_t'$ hold w.h.p. in Phase 2.2\,.

Combining \cref{thm:phase1thm}, \cref{lemma:firstPotentialSmall}, and \cref{lemma:secondPotentialSmall} gives the main result of this section.

\begin{theorem}\label{thm:easyprima}
Consider the \minority protocol starting from an arbitrary configuration in $\Omega$. Then for any constant $r>0$, there exists a time-step $\tau=\tau_1+\tau_{2,1}+\tau_{2,2}= O(w^2n \log n)$\, and a constant $C>0$ such that with probability at least $1-O(n^{-r})$ we have that
\begin{align*}
\sum_{i=1}^k\sum_{j=1}^k \left(\frac{A_i(t)}{w_i}-\frac{A_j(t)}{w_j}\right)^2 \leq Cwn\log n\,, \; \text{ and }\; \sum_{i=1}^k \sum_{j=1}^k \left( \frac{a_i(t)}{w_i}-\frac{a_j(t)}{w_j}\right)^2 \leq Cwn \log n\,,
\end{align*}
for all $t \in \{\tau,\tau+1,\ldots, \tau+n^8\}$\,.
\end{theorem}

We note that \cref{thm:easy} is a straightforward corollary of \cref{thm:easyprima}\,. The strategy to prove \cref{lemma:firstPotentialSmall} and \cref{lemma:secondPotentialSmall} 
is to show that after $O(wn)$ time-steps the value of each potential function halves (in the respective phases). In order to prove the potentials halve, we use the following lemmas that establish some properties about $\phi$ and $\psi$ respectively. 

Denote by $\mathcal F_t $ the filtration generated by the configurations of colours up to time $t$\,, and let $G_t$ be an increasing sequence of events such that $G_{t}\in \mathcal F_{t-1}$ for $t \geq 1$ and $G_0 \in \mathcal F_0$\,. The potential functions satisfy the following properties:

\begin{lemma}\label{lemma:potential1Prop}
Let $T\geq 0$ be an integer. Suppose that event $\{\xi_{t}\in \mathcal E\}$ is contained in $G_{t+1}$ for all $t\geq T$, then
there exist constants $C_1,C_2,C_3,C_4,C_5$ and $C_6$ (independent of $T$, $k$, etc.) such that for any $t\geq T$,
\begin{enumerate}
    \item \label{eqn:ExPot1}
   $\E(\ind_{G_{t+1}}\phi(t+1)|\mathcal F_t) \leq \ind_{G_{t+1}}\phi(t)\left(1-\frac{C_1}{nw} \right) + C_2$\,,
   \item \label{eqn:VarPot1}
   $\var(\ind_{G_{t+1}}\phi(t+1)|\mathcal F_t) \leq C_3 \ind_{G_t}\phi(t)/w+C_4$\,,
   \item \label{eqn:DiffPot1}$ |\ind_{G_{t+1}}\phi(t+1)-\E(\ind_{G_{t+1}}\phi(t+1)|\mathcal F_t)|\leq  C_5k\sqrt{\ind_{G_{t}}\phi(t)}+C_6k$\,.
\end{enumerate}

\end{lemma}

\begin{lemma}\label{lemma:potential2Prop}
Let $T\geq 0$ be an integer. Suppose that event $\{\xi_{t-1}\in \mathcal E\} \cap \{\psi(t-1)\geq \max(16\phi(t-1),k^2)\}$ is contained in $G_t$, then there exist constants $C_1,C_2,C_3,C_4,C_5,C_6$ (independent of $T$, $k$, etc.) such that for any $t\geq T$  
\begin{enumerate}
    \item \label{eqn:ExpectationPot2}$
  \E(\ind_{G_{t+1}}\psi(t+1)|\mathcal F_t)\leq  \ind_{G_{t}}\psi(t)\left(1-\frac{C_1}{n}\right)+  C_2$\,,
    \item \label{eqn:VarPot2}$
   \var(\ind_{G_{t+1}}\psi(t+1)|\mathcal F_t) \leq C_3 k\ind_{G_{t}}\psi(t)+C_4$\,,
    \item \label{eqn:diffPot2}
    $|\ind_{G_{t+1}}\psi(t+1)-\E(\ind_{G_{t+1}}\psi(t+1)|\mathcal F_t)|\leq C_5k\sqrt{\ind_{G_{t}}\phi(t)}+C_6k$\,.
\end{enumerate}
\end{lemma}

We remark that the only difference between the previous two results is that the second lemma requires an extra condition on the event $G_t$\,. This difference is key, as any meaningful application of  \cref{lemma:potential2Prop} requires knowing that $\psi(t-1)\geq 16\phi(t-1)$\,, and thus $\phi(t-1)$ needs to be already small for any simple application of it. 

The previous two lemmas are not particularly useful by themselves, however we prove the following concentration inequality, which combining with the previous results leads to the proof the main results of this phase.

\begin{lemma}\label{lemma:processConcentration} Let $M(t)$ be a stochastic process adapted to a filtration $\mathcal F_t$\,. Suppose that $M(t)\geq 0$ satisfies
\begin{enumerate}[(i)]
    \item \label{lemma:pM5i} $\E(M(t)|\mathcal F_{t-1}) \leq (1-\alpha)M(t-1)+\beta$, with $0<\alpha<1$, and $\beta>0$\,;
    \item \label{lemma:pM5ii} $|\E(M(t)|\mathcal F_{t-1})-M(t)|\leq \gamma$\,;
    \item \label{lemma:pM5iii} $\var(M(t)|\mathcal F_{t-1})\leq \delta^2$\,.
\end{enumerate}
Then for any $\lambda > 0$\,,
\begin{align}
     \Prob(M(t) \geq \E M(t)+\lambda) \leq \exp\left(-\frac{\lambda^2/2}{\frac{\delta^2}{(2\alpha-\alpha^2)}+\frac{\lambda \gamma}{3}}\right)\label{eqn:lemmaConcentrationFinal}\,.
\end{align}
\end{lemma}
\cref{lemma:potential1Prop} and \cref{lemma:potential2Prop} are proved in \cref{sec:proofphase2pote}. The proofs of \cref{lemma:firstPotentialSmall} and \cref{lemma:secondPotentialSmall} are deferred to \cref{sec:phase2PoteSmall}, and finally the proof of \cref{lemma:processConcentration} can be found in \cref{sec:concentrationproof}.

\subsection{Phase 3: A finer Equilibrium}\label{sec:props}
Our next goal is to prove the next theorem.
\begin{theorem}\label{thm:fair}
For constant $k$ and $w$, the \minority protocol achieves fairness.
\end{theorem}
In order to demonstrate fairness, we need to prove a finer characterisation of the colour distribution for large times. Such characterisation can only be proved by using the properties that are proved to hold after the first two phases of the process. Therefore, this new characterisation is proved in Phase 3, which starts immediately after phase  2 and last for $\tau_3$ time steps.

\begin{theorem}\label{thm:approximation}
Given $r>9$ there exists $C>0$ and $\tau = \tau_1+\tau_{2,1}+\tau_{2,2}+\tau_3 = O(w^2n\log n)$ such that with high probability $1-O(n^{-r})$, it holds that for all $i \in \{1,\ldots, k\}$ and for all $t \in \{\tau, \ldots, n^8\}$\,. 
\begin{align*}
\left|A_i(t) - \frac{w_i}{1+w} n\right| \leq C n^{3/4}(\log n)^{1/4}\,,\; \text{ and }\;  \left|a_i(t)- \frac{w_i}{(1+w)w} n\right| \leq C n^{3/4}(\log n)^{1/4} \,.
\end{align*}

\end{theorem}

As usual, for the proof of the theorem it will be convenient to restart time after phase 2. Recall that $B'_t = \{\xi(s) \in \mathcal E' \,, \forall s\in \{0,\ldots, t-1\}\}$, where $\mathcal E'$ is defined in~\cref{eqn:defiE'}\,. For the proof of the \cref{thm:approximation} we define the set $\widehat{\mathcal E} \subseteq \Omega$ as
\begin{align*}
  \widehat{\mathcal E} &= \left\{\xi \in \Omega: \sum_{i=1}^k \sum_{j=1}^k \left(\frac{a_i}{w_i} -\frac{a_j}{w_j}\right)^2 \leq C'wn\log n\right\} \cap 
  \left\{\xi \in \Omega:\sum_{i=1}^k \sum_{j=1}^k \left(\frac{A_i}{w_i} -\frac{A_j}{w_j}\right)^2 \leq C'wn\log n\right\}
\end{align*}
where $C'>0$ is a large constant. By \cref{thm:easyprima}, we have that $\xi(t) \in \widehat{\mathcal E}$ for all $t$ after Phase 2 for at least $n^8$ time steps with probability at least $1-O(n^{-r})$ (the constant $C'$ depends on $r$). Define $\widehat B_t = \{\xi(s) \in \widehat{\mathcal E}, \forall s\in \{0,\ldots,t-1\}\}$.

The proof of \cref{thm:approximation} follows by combining \cref{thm:easyprima} and the following lemma.

\begin{lemma}\label{lemma:thirdPotentialSmall}
Given $r>9$ there exists $\widehat C>0$ and $\tau_3 = O(wn\log n)$ such that with probability at least $n^{9-r}$\,, we have for all $t \in \{0,\ldots, n^8-\tau_3\}$\,,
\begin{align}
    \ind_{\widehat B_t}\left(\frac{A(t)}{w}-a(t)\right)^2 \leq \widehat Cn^{3/2}\sqrt{\log n} \,.
\end{align} 
\end{lemma}

To see why \cref{thm:approximation} follows, we recall that \cref{thm:easyprima} tells us that proportions are respective between the number of vertices representing colours \textbf{of the same shade}, however, we do not know how many vertices there are of each shade. Fortunately, \cref{lemma:thirdPotentialSmall} tells us that if we wait $\tau_3$ extra steps after the time \cref{thm:easyprima} starts working (given by $\tau_1+\tau_{2,1}+\tau_{2,2}$), then \cref{lemma:thirdPotentialSmall} provides one extra equation, which is enough to completely determine the number of agents supporting each dark colour and light colour, up to a small additive error.

The proof of \cref{lemma:thirdPotentialSmall} is provided in \cref{sec:proofFairness}, which follows the same steps of the proof of \cref{lemma:firstPotentialSmall} and \cref{lemma:secondPotentialSmall}.

\subsection{Proof of \cref{thm:fair} - Via Markov Chain approximation}\label{sec:fair}

Recall that each agent has one of $k$ colours, and each colour is either light or dark shaded. We denote by $D_i$ the dark shaded version of colour $i$, and $L_i$ its light shaded version. Therefore, in our protocol, we can imagine that each agent is moving in the state space $\{D_1,\ldots, D_k, L_1, \ldots, L_k\}$ according to some transition rule that depends on the current state of the agent as well as the states of all other agents (i.e. the trajectory of one agent is not a Markov chain by itself). Let $M^o(t)$ denote the trajectory of a particular agent on the state space $\{D_1, D_2, \dots, D_k,L_1,L_2, \dots, L_k\}$. 

Even though $M^o$ is not a Markov chain, we can approximate it by one. Indeed, let $M$ be the Markov chain on the state space $\{D_1, D_2, \dots, D_k,L_1, \dots, L_k\}$  with transition matrix $P$ given by
\begin{align*}
    &P(L_j,D_i) =\frac{w_i}{(1+w)n}\,, \qquad \text{for all $i,j$} \,;\\
    &P(L_i,L_i) = 1- \frac{w}{(1+w)n}\,, \quad \text{for all $i$} \,;\\
    &P(D_i,L_i) =\frac{1}{ (1+w)n}\,, \qquad \text{for all $i$} \,;\\
    &P(D_i,D_i) = 1-\frac{1}{(1+w)n}\,, \quad \text{for all $i$}\,,
\end{align*}
while all other entries of $P$ are 0. We introduce $M$ as it describes the trajectory of a particle across all the states of the system when the system is perfectly balanced: with probability $1/n$ the agent is chosen to observe a random neighbour, and in perfect equilibrium, each dark shaded colour appears in a proportion $w_i/(1+w)$, and each light shaded colour in proportion $w_i/(w(1+w))$ (c.f. \cref{thm:approximation}). The idea is that, despite the fact that such perfect equilibrium is very unlikely to be achieved, \cref{thm:approximation} ensures that the system stays close to such equilibrium for long periods of time.

We claim that the stationary distribution of $M$ is given by \[\pi(L_i) = \frac{w_i/w}{1+w}\,,\; \text{ and }\;\pi(D_i) = \frac{w_i}{1+w}\,, \]
which can be directly verified by definition of stationary distribution, that is,
\begin{align}
    &\pi(D_i) = \pi(D_i)  P(D_i,D_i)+\sum_{j=1}^k \pi(L_j)  P(L_j, D_i)\,, \quad \text{and}\nonumber\\ \\
    &\pi(L_i) = \pi(L_i)P(L_i,L_i)+\pi(D_i)P(D_i,L_i)\,.\label{eqn:stationarydistribapprox}
\end{align}

Now, we use the transition matrix $P$ to approximate the trajectory of $M^o$, indeed, let $err = c(\log n)^{1/4}/ n^{1/4}$, for some large enough constant $c$.
Then by \cref{thm:approximation}, for any $t\in \{\tau, \dots, n^8\}$ we have, that the following holds w.p. $1-\delta$ where $\delta=O(n^{-r})$.
\begin{align}
\Prob(M^o(t+1) = L_i| M^o(t) = D_i,\mathcal F_t)&= \frac{A_i(t)/w_i}{n(n-1)} = P(D_i,L_i) \pm err \label{eqn:approxMo1}\,,
\end{align}
since with probability $1/n$ the agent is chosen, and it changes its colour from dark to light with probability $A_i(t)/(w_i(n-1))$, by the definition of the protocol. In a similar fashion, we have
\begin{align}
\Prob(M^o(t+1) = D_i| M^o(t) = D_i,\mathcal F_t)&= P(D_i,D_i)\pm err  \nonumber\\ 
\Prob(M^o(t+1) = D_j| M^o(t) = L_i,\mathcal F_t)&= P(L_i,D_j)\pm err,\quad \text{for all $i,j$}\nonumber\\
\Prob(M^o(t+1) = L_i| M^o(t) = L_i,\mathcal F_t)&= P(L_i,L_i)\pm err\,. \label{eqn:approxMo2}
\end{align}
Note that other transitions have probability 0 as they are impossible by design of the protocol (they have also probability 0 in the transition matrix $P$).

Now, fix some state, say $D_\ell$, and define the following transition matrix $P_{D_\ell}^+$, given by
\begin{align*}
P_{D_\ell}^+(D_{\ell},L_{\ell}) &= P(D_{\ell},L_{\ell})-err \,,\\
P_{D_\ell}^+(D_{\ell},D_{\ell}) &= P(D_{\ell},D_{\ell})+err \,,\\
P_{D_\ell}^+(D_i,L_i)&= P(D_i,L_i)+err \,,\quad \text{for $i\neq \ell$}\,, \\ 
P_{D_{\ell}}^+(D_i,D_i)&= P(D_i,D_i)-err \,,\quad \text{for $i\neq \ell$}\,, \\ 
P_{D_{\ell}}^+(L_{i},D_{\ell})&= P(L_i,D_{\ell})+k err \,, \quad\text{for all $i$}\,,\\
P_{D_{\ell}}^+(L_i,D_j)&= P(L_i,D_j)- err\,, \quad \text{for $i$, and $j \neq \ell$}\,, \\ 
P_{D_{\ell}}^+(L_i,L_i)&= P(L_i,L_{i})- err\,, \quad \text{for all $i$}\,. 
\end{align*}

The idea is that $P^+_{D_{\ell}}$ is a transition matrix, such that it increases the probability of all the transitions of $M^o$ that move the agent closer to the state $D_{\ell}$ and decreases the probabilities of the transitions of $M^o$ that prevent the agent from getting closer to $D_{\ell}$. In a similar way we can define $P^-_{D_{\ell}}$, which is defined the same way as $P^+_{D_{\ell}}$ but changing the sign for the $err$ term, and thus representing a transition matrix that decreases transition probabilities that move the agent closer to $D_{\ell}$, and increases the others. Similarly, we define $P^+_{L_{\ell}}$ and $P^-_{L_{\ell}}$.

Now, it is convenient to restart time at time $\tau$ in which the approximation of \cref{thm:approximation} starts holding true. Recall that the even of \cref{thm:approximation} holds with probability $1-\delta$, where $\delta = O(n^{-r})$ for fixed $r>9$. Let $N^o_{D_{\ell}}(t)$ be the number of times $M^o$ has been in state $D_{\ell}$ up to time $t$, and similarly $N^+_{D_{\ell}}(t)$ the number of times a particle moving according to $P^+_{D_{\ell}}$ hits $D_{\ell}$ up to time $t$, and  define $N^-_{D_{\ell}}$ similarly. By the construction of the transitions matrices $P^+_{D_{\ell}}$ and $P^-_{D_{\ell}}$, w.p at least $1-\delta$ the following majorisation between the random variables holds:
\begin{align*}
N^-_{D_{\ell}}(t) \precsim N^o_{D_{\ell}}(t) \precsim N^+_{D_{\ell}}(t)
\end{align*}
for all $t\leq n^8-\tau$. 
To see this, 
note that starting at any state $x\neq D_\ell$ one can use a coupling to show that the time it takes for the process $M^o$ to reach $D_\ell$ from $x$ is at most the time it takes the chain $P^+_{D_{\ell}}$ to reach $D_\ell$ from $x$. A similar argument can be made for the time the chain stays put at $D_\ell$. With the same argument we can see that $N^-_{D_{\ell}}(t) \precsim N^o_{D_{\ell}}(t)$.

In the low-probability event (w.p. $\delta$) that the event of \cref{thm:approximation} does not hold, we consider the simple bound $0 \leq N^o_{D_{\ell}}(n^8) \leq n^8$.

Note that $P^+_{D_{\ell}}$  is a recurrent and ergodic Markov chain with (unique) stationary distribution denoted by $\pi^+$. Note that the state space has size $2k$ which is finite, and thus the mixing time is finite as well. Therefore, by the Chernoff's bounds for Markov chains (see \cref{thm:MCchernoff}), given $r>0$ there exists $c>0$ such that
\begin{align*}
|N^+_{D_{\ell}}(t)- \pi^+(D_{\ell})t|\leq c\sqrt{\pi^+(D_{\ell})t \log n}
\end{align*}
with probability $1-\delta$.

Finally, we notice that $P^+_{D_{\ell}}$ is just a small perturbation of $P$, and it can be easily verified that $\pi^+(D_{\ell}) = \pi(D_{\ell})+ O(err)$ (see  \cref{eqn:stationarydistribapprox}), therefore
\begin{align*}
 N^+_{D_{\ell}}(t) \leq \frac{w_{\ell}t}{1+w}+O\left(\sqrt{\frac{w_{\ell}t\log n}{1+w}}+ (err\cdot t)\right)
\end{align*}
with probability at least $p=1-\delta$ for arbitrary $r>0$. The same argument applies to $N^-_{D_{\ell}}$.

We conclude the proof by noting that from the beginning of the process, the approximations in \cref{eqn:approxMo1} and \cref{eqn:approxMo2} only hold from time $\tau = O(wn\log n)$ as stated in \cref{thm:approximation}, and such approximation is valid up to time $n^8$ with polynomially large probability. Therefore, we still need to count the number of hits to $D_{\ell}$ in the interval $[0,\tau]$, however, since $\tau = O(n\log n)$, even the worst-case bounds in such interval are second-order terms compared to the hits in the interval $[\tau,n^8]$

We deduce that with probability at least $1-O(\delta)$ we have that for a particular agent $u$, 
\begin{align*}
\frac{\left|\{  t \in[0, n^8] \colon c_{u}(t')=\ell, b_u(t) = 1\}\right|}{n^8}
    =\frac{w_{\ell}}{1+w}(1\pm o(1))\,,
\end{align*}
and similarly, by repeating the same argument for $L_{\ell}$, we get, with the same probability that
\begin{align*}
\frac{\left|\{  t \in[0, n^8] \colon c_{u}(t')=\ell, b_u(t) = 0\}\right|}{n^8}
    =\frac{w_{\ell}}{(1+w)w}(1\pm o(1))\,,
\end{align*}
and by the union bound, with probability at least $1-O(\delta)$ it holds that
\begin{align*}
\frac{\left|\{  t \in[0, n^8] \colon c_{u}(t')=\ell\}\right|}{n^8}
    =\frac{w_{\ell}}{w}(1\pm o(1))\,,
\end{align*}
and finally the same holds for all $i \in \{1,\ldots, k\}$ and all agents at the same time by the union bound, with probability $1-O(nk\delta)$. Note that we just prove the approximation up to time $n^8$, but for larger times we divide time in segments of length $n^8$, and in each segment we allow $\tau$ time-steps in order to reach the approximation of \cref{thm:approximation}, and then we apply the previous domination argument. Note that the probability that we have a successful approximation is $1-O(kn\delta)=1-O(n^{-r+1})$ in each interval where $r$ is arbitrarily large.
In case something the event of \cref{thm:approximation} does not hold, we just bound the number of visits by $n^8$ and $0$, respectively. Choosing $r > 10$ ensures that w.h.p. fairness is achieved in the interval.
Using Chernoff's bounds, one can also show fairness over longer intervals.

\section{Future work and conclusions}
We have introduced the notion of  diversity, fairness and sustainability to the realm of population protocols. We showed that the simple \minority protocol achieves those properties.

Natural open problems are  to analyse the derandomised version of our protocol, and  to extend our analysis to $k$ and $w$ depending on the number of agents $n$. Another research direction is to find protocols that attain a stronger notion of diversity, where the error term in \cref{eqn:defiDiver} is much smaller than $\tilde O(1/\sqrt{n})$. We think that finding such protocol will lead to very interesting trade-off between time and space complexity. Another line of research is to investigate the diversification protocol in different graph topologies, other than the complete graph.

An interesting research direction is to describe what lies in between consensus and diversification? Finally, one can ask whether our protocol emerges naturally in the context of learning, or in other distributed systems, perhaps even in biology.

\section{Proofs}\label{sec:proofs}
Recall that we assume that the $w_i,i\in[k]$ and $k$ are constants.

\subsection{Proof of \cref{sec:phase1}} \label{sec:proofsphase1}
\begin{proof}[Proof of \cref{lemma:part1}]
To begin with, note that 
$$\sum_{i=1}^k \frac{A_i^2}{w_i} = w \sum_{i=1}^k \left(\frac{A_i}{w_i}\right)^2 \frac{w_i}{w}\geq w\left( \sum_{i=1}^k \frac{A_i}{w_i} \frac{w_i}{w}\right)^2 = \frac{A^2}{w} \,,$$
where in the first inequality we used the Jensen's inequality since $w_i/w>0$ and $\sum_{i=1}^k w_i/w = 1$\,.

Let $p$ be the probability of increasing $a$ and $q$ the probability of decreasing it.
Using the above inequality, we get
\begin{align}
p &= \sum_{i=1}^k \frac{A_i (A_i -1)}{n (n-1)w_i} \geq \frac{A^2}{wn(n-1)}-\frac{A}{n(n-1)} \,, \text{ and} \nonumber\\
q &= \sum_{i=1}^k \frac{a_i}{n} \frac{A}{n-1} = \frac{aA}{n^2 - n}\,. \label{eqn:valuesofpandq}
\end{align}

W.l.o.g. assume that the current configuration is not in $R_1$\,, we must have that $a < n(1-\varepsilon) \frac{1}{w+1}$ and therefore, $A \geq n- n(1-\varepsilon) \frac{1}{w+1}$\,. Thus, 
\[ 
\frac{A}{a} \geq \frac{1-\frac{1-\varepsilon}{w+1}}{ \frac{1-\varepsilon}{1+w}} = \frac{1+w-1+\varepsilon}{1-\varepsilon} = \frac{w+\varepsilon}{1-\varepsilon}\,.
\]

Then, \cref{eqn:valuesofpandq} yields that the probability to decrease in an active time-step (i.e. when $a$ changes its value) is
\begin{align*}
\frac{q}{q+p} &\leq \frac{aA}{aA+A^2/w -A} \leq  \frac{w}{w+A/a}
\leq\frac{1}{2}-\frac{w+\varepsilon -2(1-\varepsilon)w}{2(w+\varepsilon)} \leq \frac{1}{2} -\frac{\varepsilon}{3} \,.
\end{align*}

This yields a biased random walk, that maintain the probability of increasing at least $1/2+\eps/3$ until $\xi(t) \in R_1$\,.
Consider only active steps, where $a(t)$ either increases or decreases its value, and define the stopping time $T^*=\min_{\text{active step $t$}}\{\xi(t) \in R_1 \}$\,, and an interval of active time-steps of length $\ell$\,.
Hence, according to the Azuma-Hoeffding inequality, for the number of increases $I$ minus the number of decreases $D$ we have
\[ \Pr{ I - D\leq \frac{\ell \eps}{5} } \leq \exp\left(- \frac{\left(\frac{\ell \eps}{5}\right)^2}{2\ell} \right) \,.\]

Choosing $\ell= 10 n/\eps$\,, it guarantees that $T^* \leq \ell$ w.p. at least $1-\exp\left(- \frac{\eps n}{50} \right)$\,.
It remains to show the relationship between active time-steps and regular time-steps.
Recall, 
\[ p \geq \frac{A^2/w-A}{n^2-n}\geq \frac{A^2-Aw}{wn^2} =  \Omega(1/w) \,.\]

Hence, there is only a factor $w$ difference throughout the entire regime until $T^*$\,. This yields the first part of the claim.

For the second part, we consider the setting where the biased random walk starts at
$n(1-3\varepsilon/2) \frac{1}{w+1}$ and increases w.p. at least $1/2+\eps/3$\,.
We will apply \cref{pro:CaminataAleatoriaParcial} to determine the probability of hitting $n(1-\varepsilon) \frac{1}{w+1}$ before hitting $n(1-2\varepsilon) \frac{1}{w+1}$\,.
We define $Z_t = a_t - (1-3\varepsilon/2) \frac{n}{w+1}+(\varepsilon/2) \frac{n}{w+1}$,
$b= \varepsilon \frac{n}{w+1}$, $s=\frac{\varepsilon}{2} \frac{n}{w+1}$.
Let $T= \min \{ t\geq 0 ~|~ Z_t \in \{ 0,b \}\}$. Then, applying
\cref{pro:CaminataAleatoriaParcial} and using that $\frac{y-x}{1-x} \leq y$, we get
\begin{align*}
\Pr{Z_T=0} &= \frac{ \left(\frac{1/2-\eps/3}{1/2+\eps/3}\right)^{s}- \left(\frac{1/2-\eps/3}{1/2+\eps/3}\right)^{b}   }{ 1- \left(\frac{1/2-\eps/3}{1/2+\eps/3}\right)^{b}}\leq \left(\frac{1/2-\eps/3}{1/2}\right)^{s} \leq 
\left(1-2\eps/3\right)^{s} \\
&\leq \left(1-2\eps/3\right)^{\frac{\varepsilon}{2} \frac{n}{w+1}} =\exp\left(-2c\frac{n\eps^2}{w} \right) \,,
\end{align*}
for $c>0$ small enough.

By taking union bound over $\exp(cn\eps^2/w)$ time-steps yields that $T'_1 \geq \exp(cn\eps^2/w)$ w.p. at least $1-\exp\left(-c\frac{n\eps^2}{w} \right)$\,.
This completes the proof.

\end{proof}

\begin{proof}[Proof of \cref{lemma:part2}]
Let $p$ be the probability of increasing $A_i$ and $q$ the probability of decreasing it. We have
\begin{align*}
    &p = \sum_{j=1}^k \frac{a_j}{n}\frac{A_i}{n-1} = \frac{a A_i}{n^2-n} \,, \text{ and}\\
    &q = \frac{A_i}{n}\frac{(A_i-1)}{n-1} \frac{1}{w_i} \,.
\end{align*}

Assume that $A_i$ does not satisfy the property of region $R_2$ and note that $a$ is in region $S_1$ such that $a \geq (1-2\eps)n/(w+1)$\,, then 
\[ \frac{A_i}{aw_i} \leq \frac{\frac{(1-3\eps)w_i}{1+w}}{\frac{(1-2\eps)}{1+w} w_i} = \frac{1-3\eps}{1-2\eps} \,.\]

Again, only considering active steps, we get for $\eps \leq 2/5$\,,
\[ \frac{p}{q+p} 
= \frac{a}{\frac{A_i-1}{w_i}+ a} 
\geq \frac{1}{1+\frac{A_i}{aw_i} }
\geq
 \frac{1-2\eps}{1-2\eps+ 1-3\eps 
 }\geq\frac12 + \frac{\eps}{4}\,.
\]

Using a similar biased random work analysis on the active steps as in the proof of \cref{lemma:part1}\,, we get that in $O(n/\eps)$ active steps $A_i$ has the desired size with overwhelming probability. 

Note that the ratio of active steps is proportional to
$p$, which is up to constants due to our lower bound on $a$ and $ A_i/(wn) $\,. Thus increasing $A_i$ by $1$ takes $O(\frac{n w}{A_i\eps})$ time-steps. 
Hence, the total time to raise $A_i$ all the way up to linear in $n w_i/w$ is dominated by the geometric series
$\sum_{i=0}^{\log n/2} 2^i \cdot O(\frac{n w}{2^i\eps})=O(\frac{w n \log n }{\eps})$\,.

Finally, similar to \cref{lemma:part1}\,, $A_i$ remains within the stated bounds with overwhelming probability.
\end{proof}

\subsection{Proofs of \cref{sec:phase2}}
\label{sec:proofsphase2}
\subsubsection{Proof of \cref{lemma:potential1Prop} and~\cref{lemma:potential2Prop}}\label{sec:proofphase2pote}

We proceed to prove \cref{lemma:potential1Prop} and~\cref{lemma:potential2Prop}\,. In the proofs we abuse notation and use lower-case letter such as $c$, $c_1$, $c_2$, etc. to denote constants, and for convenience \textbf{ we reuse the same name for different constants}.

\begin{proof}[Proof of \cref{lemma:potential1Prop}]
We first introduce some notation. For simplicity, let $A_i = A_i(t)$ and $A_i' = A_i(t+1)$\,. We define the following variables $d_i = \frac{A_i'}{w_i} - \frac{A_i}{w_i}$\,, and $q_i = \frac{A_i}{w_i}$\,. Let $D_{r} = \sum_{i=1}^k d_i^r$\,, and $Q_{r} = \sum_{i=1}^kq_i^r$\,, for $r=\{1,2, \cdots\}\,.$ Finally, denote $q_{ij} = q_i-q_j$ and $d_{ij} = d_i-d_j$.

With the above notations, the potential function $\phi$ defined in~\cref{eqn:potential1Defi} becomes $\phi(t) = \sum_{i=1}^k \sum_{j=1}^k q_{ij}^2$\,, and we obtain
\begin{align} \label{eq:phi'}
    \ind_{G_{t+1}} \phi(t+1) = \ind_{G_{t+1}}\sum_{i=1}^k \sum_{j=1}^k (q_{ij} + d_{ij})^2 = \ind_{G_{t+1}} \left( \phi(t) + 2 \sum_{i=1}^k \sum_{j=1}^k q_{ij} d_{ij} + \sum_{i=1}^k \sum_{j=1}^k d_{ij}^2 \right) \,.
\end{align}

Note that $d_i$ is the proportional change of the value $A_i$ in one step, and that $A_i$ increases its value by 1 with probability $\frac{aA_i}{n(n-1)}$ and decreases by 1 with probability $\frac{A_i(A_i-1)}{w_in(n-1)}$\,. Then,
\begin{align} \label{eq:E_di}
    \E(d_i|\mathcal F_t) &= \frac{1}{w_i}\left( \frac{aA_i}{n(n-1)} - \frac{A_i(A_i-1)}{n(n-1) w_i} \right) = \frac{1}{n(n-1)} \left( a q_i - q_i^2 + \frac{q_i}{w_i} \right) \,, \text{ and} \nonumber\\
    \E(d_i^2|\mathcal F_t) &= \frac{1}{w_i}\left( \frac{aA_i}{n(n-1)} + \frac{A_i(A_i-1)}{n(n-1) w_i} \right) = \frac{1}{n(n-1)} \left( a q_i + q_i^2 - \frac{q_i}{w_i} \right) \,.
\end{align}

Notice that $\left|\E(d_i | \mathcal F_t)- (aq_i-q_i^2)/n^2 \right| \leq 3/n$\,. Then we have the following bound on the term $\sum_{i=1}^k \sum_{j=1}^k q_{ij} d_{ij}$ in~\cref{eq:phi'}\,. For some constant $c>0\,,$
\begin{align*}
    \ind_{G_{t+1}}\E\left(\sum_{i=1}^k \sum_{j=1}^k q_{ij}d_{ij} \given \mathcal F_t\right) &\leq \ind_{G_{t+1}}\sum_{i=1}^k \sum_{j=1}^k \frac{ q_{ij}}{n^2} \left[ a q_i - q_i^2 - (a q_j - q_j^2 ) \right] + \sum_{i=1}^k \sum_{j=1}^k|q_{ij}|\frac{3}{n} \\
    &= \ind_{G_{t+1}}\sum_{i=1}^k \sum_{j=1}^k q_{ij}^2 \frac{a-(q_i + q_j)}{n^2} + c\leq - \ind_{G_{t+1}} \frac{1-3\delta}{(1+w)n} \phi(t) +c\,,
\end{align*}
where in the last step we used the facts that $\{\xi(t)\in \mathcal E\}\subseteq G_{t+1}$\,. By the definition of $\mathcal E$ in \cref{eqn:defi_E} we have $a\leq (1+\delta)n/w$ and $q_i \geq (1-\delta)n/w$\,, where $\delta$ is sufficiently small.

As for the term $\sum_{i=1}^k \sum_{j=1}^k d_{ij}^2$ in~\cref{eq:phi'}\,, we notice that $d_i d_j = 0$ for $i \neq j$\,, which yields
\begin{align*}
    \ind_{G_{t+1}}\E\left(\sum_{i=1}^k \sum_{j=1}^kd_{ij}^2 \given \mathcal F_t \right) &= \ind_{G_{t+1}}2(k-1) \E\left( \sum_{i=1}^k d_i^2 \given \mathcal F_t \right) \leq \ind_{G_{t+1}} 2(k-1) \sum_{i=1}^k \frac{ a q_i + q_i^2}{n^2} \leq \ind_{G_{t+1}} \frac{4k^2 (1+\delta)^2}{(1+w)^2} \,,
\end{align*}
where in the last step we apply that $a\leq (1+\delta)n/w$ and $q_i \leq (1+\delta)n/w$ for sufficiently small $\delta$ since $\{\xi(t)\in \mathcal E\}\subseteq G_{t+1}$\,.

Notice that $k \leq w$\,. Thus, replacing~\cref{eq:phi'} with the above calculations and bounds yields, for some constants $C_1, C_2 >0\,,$
\begin{align*}
    \ind_{G_{t+1}}\phi(t+1) &= \ind_{G_{t+1}} \left( \phi(t) - \frac{2(1-3\delta)}{(1+w)n} \phi(t)+c + \frac{4k^2 (1+\delta)^2}{(1+w)^2}  \right) \leq \ind_{G_{t+1}}\phi(t)\left(1-\frac{C_1}{wn}\right)+C_2 \,.
\end{align*}

We continue to prove \cref{eqn:VarPot1}. We start by noting that \begin{align*}
\phi(t) = &\sum_{i=1}^k \sum_{j=1}^k (q_i-q_j)^2 = 2kQ_2-2Q_1^2,
\end{align*}
and
\begin{align*}
\phi(t+1) &= 2k\sum_{i=1}^k (q_i+d_i)^2-2(Q_1+D_1)^2 \\
&= 2k\left(Q_2+D_2+2\sum_{i=1}^k q_id_i\right) - 2Q_1^2-2D_1^2 - 4Q_1D_1 \\
&= \phi(t) + 4k\sum_{i=1}^k q_id_i-4Q_1D_1 +(2k-2)D_2 \,.
\end{align*}

We compute $\var(\ind_{G_{t+1}}\phi(t+1)-\ind_{G_t}\phi(t)|\mathcal F_t)$ by using the previous equation. First note that if $\ind_{G_{t+1}}=0$ then $\var(\ind_{G_{t+1}}\phi(t+1)-\ind_{G_{t}}\phi|\mathcal F_t) = \var(\ind_{G_{t}}\phi(t)|\mathcal F_t) = 0$, since $\ind_{G_{t}}\phi(t) \in \mathcal F_t$ and $\ind_{G_{t+1}}\in \mathcal F_t$. If $\Ind_{G_{t+1}} = 1$, then $\ind_{G_t}=1$ and thus
\begin{align}\label{eq:var_deltaPhi}
\var(\ind_{G_{t+1}}&\phi(t+1) -\ind_{G_t}\phi(t)|\mathcal F_t) = \ind_{G_{t+1}} \var( \phi(t+1)-\phi(t)|\mathcal F_t) \nonumber\\
&= \ind_{G_{t+1}}\var\left(\sum_{i=1}^k(4kq_i-4Q_1)d_i + (2k-2)D_2 \given \mathcal F_t\right) \nonumber\\
&\leq \ind_{G_{t+1}}\E\left(\left[\sum_{i=1}^k(4kq_i-4Q_1)d_i + (2k-2)D_2\right]^2 \given \mathcal F_t\right) \nonumber\\
&\leq  2\cdot \ind_{G_{t+1}}\E\left(\left[\sum_{i=1}^k(4kq_i-4Q_1)d_i \right]^2 \given \mathcal F_t\right) + 2\cdot \ind_{G_{t+1}}\E\left(\left[(2k-2)D_2\right]^2\given \mathcal F_t\right)
\end{align}

We upper bound each of the terms in the equation above as follows: By~\cref{eq:E_di} and the fact that $a\leq (1+\delta)n/w$ and $q_i \leq (1+\delta)n/w$ for sufficiently small $\delta$ since $\{\xi(t)\in \mathcal E\}\subseteq G_{t+1}$\,, we know that $\E(d_i) \le \frac{4 \delta^2}{(1+w)^2}$ and $\E(d_i^2) \leq \frac{2(1+\delta)^2}{(1+w)^2} \,.$ Also notice that $\E(d_i^{2r+1}) = \E(d_i)$ and $\E(d_i^{4}) \leq \E(d_i^2)$. Then for term at the right in \cref{eq:var_deltaPhi} we have, 
\begin{align*}
    2\cdot\ind_{G_{t+1}} \E \left(\left[(2k-2)D_2\right]^2\given \mathcal F_t\right) = 8\cdot\ind_{G_{t+1}} (k-1)^2 \sum_{i=1}^k \E \left( d_i^4 \given \mathcal F_t \right) \leq 8 \cdot \ind_{G_{t+1}} k^2 \frac{2(1+\delta)^2k}{(1+w)^2} \leq 8(1+\delta)^2 k = C_4 k \,.
\end{align*}
For the term at the left in \cref{eq:var_deltaPhi}, we notice that $d_id_j = 0$ for $i\neq j$, and $\phi(t) = 2kQ_2-2Q_1^2$\,, and then
\begin{align*}
2\cdot \ind_{G_{t+1}}\E &\left(\left[\sum_{i=1}^k(4kq_i-4Q_1)d_i \right]^2 \given \mathcal F_t\right) = 2\cdot \ind_{G_{t+1}} \sum_{i=1}^k (4kq_i-4Q_1)^2 \E(d_i^2|\mathcal F_t) \\
&= 16\cdot \ind_{G_{t+1}} 2(k^2 Q_2 - kQ_1^2) \E(d_i^2|\mathcal F_t) \leq 16\cdot \ind_{G_{t+1}} k \phi(t) \frac{2(1+\delta)^2}{(1+w)^2} \leq \ind_{G_{t+1}} \frac{C_3}{w} \phi(t) \,.
\end{align*}

Hence, substituting~\cref{eq:var_deltaPhi} with the above calculations yields for $C_3, C_4 >0$
\begin{align*}
    \var(\ind_{G_{t+1}}\phi(t+1)-\ind_{G_t}\phi(t)|\mathcal F_t) = \ind_{G_{t+1}}\frac{C_3}{w} \phi(t) + C_4 k \,.
\end{align*}

Finally, for \cref{eqn:DiffPot1}, recall that $q_{ij} = q_i-q_j$ and $d_{ij} = d_i-d_j$, and then from \cref{eq:phi'} we get
\begin{align*}
|\ind_{G_{t+1}}\phi(t+1)-\E(\ind_{G_{t+1}}\phi(t+1)|\mathcal F_t)| &=  \ind_{G_{t+1}}\left|\sum_{i=1}^k \sum_{j=1}^k\left[2q_{ij}(d_{ij}-\E(d_{ij}|\mathcal F_t))+d_{ij}^2-\E(d_{ij}^2|\mathcal F_t) \right] \right|\nonumber\\
&\leq \ind_{G_{t+1}}\left(2\sqrt{\sum_{i=1}^k \sum_{j=1}^k q_{ij}^2} \sqrt{\sum_{i=1}^k \sum_{j=1}^k(d_{ij}-\E(d_{ij}|\mathcal F_t))^2}+ \sum_{i=1}^k \sum_{j=1}^kd_{ij}^2+\E(d_{ij}^2|\mathcal F_t)\right)\nonumber\\
&\leq 2\ind_{G_t}\phi(t)\sqrt{16k}+4k \,,
\end{align*}
where in the first inequality we used the Cauchy-Schwarz inequality, and the fact that $|d_{ij}|\leq 1$\,.
\end{proof}

We show \cref{lemma:potential2Prop} using a similar idea as the proof of \cref{lemma:potential1Prop}\,.

\begin{proof}[Proof of \cref{lemma:potential2Prop}]
Let $a_i = a_i(t)$, $a_i' = a_i(t+1)$, $\bar q_i = a_i/w_i$, $\bar d_i = (a_i'-a_i)/w_i$. Define $\bar Q_j = \sum_{i=1}^k (\bar q_i)^j$ and $\bar D_j = \sum_{i=1}^k (\bar d_i)^j$. Recall that $A = \sum_{i=1}^k A_i$ and we denote $A = A(t)$\,. We also consider all the notation defined in the proof of \cref{lemma:potential1Prop}.

We begin by enlisting some properties of $\bar d_i$. 
Since $a_i$ increases by 1 with probability $\frac{A_i(A_i-1)}{n(n-1)w_i}$ and decreases by 1 with probability $\frac{a_iA}{n(n-1)}$ we have that
\begin{align*}
    \E(\bar d_i|\mathcal F_t) &= \frac{1}{w_i} \left[\frac{A_i (A_i -1)}{n(n-1)w_i} - \frac{a_i A}{n(n-1)} \right] = \frac{1}{n^2}\left(q_i^2- \bar q_iA\right) -\frac{q_i}{w_i n^2}
    \end{align*}
    which implies that
    
   \begin{align*}\left|\E(\bar d_i|\mathcal F_t) - \frac{1}{n^2}\left(q_i^2- \bar q_iA\right)\right| \leq 3/n \,.
\end{align*}

Let's prove the first item. Denote $\bar q_{ij} = \bar q_i - \bar q_j$ and $\bar d_{ij} = \bar d_i-\bar d_j$\,. Then $\psi(t) = \sum_{i=1}^k \sum_{j=1}^k {\bar q_{ij}}^2\,.$ Note that $\bar d_i \bar d_j = 0$ for $i\neq j$ and $|\bar d_i|^r \leq |\bar d_i|\leq 1$ for $r\geq 1$\,, then
\begin{align} \label{eq:E_psi}
\E(\psi(t+1) |\mathcal F_t) &= \E\left( \sum_{i=1}^k \sum_{j=1}^k (\bar q_{ij} + \bar d_{ij})^2 \given \mathcal F_t\right) \nonumber\\
&\leq \psi(t) + 2\sum_{i=1}^k \sum_{j=1}^k \bar q_{ij} \E(\bar d_{ij} | \mathcal F_t) + 2k \sum_{i=1}^k  \E((\bar d_i)^2|\mathcal F_t) \nonumber\\
&\leq \psi(t) + \frac{2}{n^2}\sum_{i=1}^k \sum_{j=1}^k \bar q_{ij}(q_i^2-q_j^2 -A\bar q_{ij}) + \frac{6}{n} \sum_{i=1}^k \sum_{j=1}^k \left|\bar q_{ij}\right| + 2k\sum_{i=1}^k\E(|d_i||\mathcal F_t) \nonumber\\
&\leq \left(1-\frac{2A}{n^2}\right) \psi(t) + \frac{2}{n^2}\sum_{i=1}^k \sum_{j=1}^k \bar q_{ij} q_{ij}(q_i + q_j) + \frac{6}{n} \sum_{i=1}^k \sum_{j=1}^k \left|\bar q_{ij}\right| + 2k\sum_{i=1}^k \sum_{j=1}^k\left(\frac{q_i^2 + q_i A}{n^2} + \frac{3}{n}\right) \,.
\end{align}

If we work on the event $\{\xi(t) \in \mathcal E\}$, then $q_i\leq c_1n/k$ for some $c_1>0$\,. Applying the Cauchy-Schwarz inequality yields
\begin{align*}
    \sum_{i=1}^k \sum_{j=1}^k\bar q_{ij}q_{ij}(q_i+q_j) \leq \sqrt{\psi(t) \phi(t) c_1n/k},
    \end{align*}
and 
\begin{align*}\sum_{i=1}^k \sum_{j=1}^k |\bar q_{ij}| \leq k \sqrt{\psi(t)}\,.
\end{align*}
It holds that $a \leq c_2n/k$ for some $c_2 >0$ on the event $\{\xi(t)\in \mathcal E\}$\,. Then for some $c_3 >0$
\begin{align*}
    2k\sum_{i=1}^k \sum_{j=1}^k\left(\frac{q_i^2 + q_i A}{n^2} + \frac{3}{n}\right) &\leq 2k\sum_{i=1}^k \sum_{j=1}^k\left(\frac{q_i^2 + a_i A}{n^2} + \frac{3}{n}\right) = \frac{2k}{n^2} \left( \sum_{i=1}^k \sum_{j=1}^k q_i^2 + a A + 3kn\right) \\
    &\leq \frac{2k}{n^2} \left( \sum_{i=1}^k \sum_{j=1}^k \left(\frac{c_1 n}{k} \right)^2 + \frac{c_2n}{k}A + 3kn\right) \leq c_3 \,.
\end{align*}

Replacing the previous three equations in \cref{eq:E_psi} yields
\begin{align*}
\E(\ind_{G_{t+1}}\psi(t+1) |\mathcal F_t) &\leq \ind_{G_{t+1}}\psi(t)\left(1-\frac{2A}{n^2}\right) + \frac{2}{n^2}\sqrt{c_1\psi(t)\phi(t)n/k}+ \frac{6k\sqrt{\psi(t)}}{n}+ c_3 \\
&= \ind_{G_{t+1}}\psi(t)\left(1-\frac{2A- 2\sqrt{c_1n\phi(t)/(\psi(t)k)}- 6kn/\sqrt{\psi(t)}}{n^2}\right)+c_3 \\
&\leq \ind_{G_{t+1}}\phi(t)\left(1-\frac{(1+o(1))n- O(\sqrt{n\log n})- (1-o(1))n}{n^2}\right)+c_3 \\
&= \ind_{G_{t+1}}\phi(t)\left(1-\frac{C_1}{n}\right)+C_2 \,,
\end{align*}
for some $C_1,C_2>0$, where in the previous to last equation we used that in the event $G_{t+1}$ it holds that $\phi(t)\geq \max\{16\phi(t), k^2\}$, and that $2A\geq 2(1-\delta)^2nw/(w+1) \geq (1+c)n$ as $\delta$ is small enough and $w\geq k \geq 2$\,.

We continue with \cref{eqn:VarPot2}.
We note that 
\begin{align*}
\psi(t+1)-\psi(t) &= 4k\sum_{i=1}^k \bar q_i \bar d_i + 4\bar Q_1 \bar D_1+ 2k\sum_{i=1}^k (\bar d_i)^2- 2(\bar D_1)^2\nonumber\\
&=4k\sum_{i=1}^k \bar q_i \bar d_i + 4\bar Q_1 \bar D_1+ (2k-2)\bar D_2 \,.
\end{align*}
Following the same reasoning of 
\cref{eqn:VarPot1} of \cref{lemma:potential1Prop} yields

\begin{align}
\var(\ind_{G_{t+1}}\psi(t+1)-\ind_{G_t}\psi(t)|\mathcal F_t) &= \ind_{G_{t+1}}\var(\psi(t+1)-\psi(t)|\mathcal F_t)\nonumber\\
&\leq 2\ind_{G_{t+1}} \left(\sum_{i=1}^k(4k\bar q_i-4\bar Q_1)^2\var(\bar d_i|\mathcal F_t)+\E\left(\left[(2k-2)\bar D_2\right]^2\given \mathcal F_t\right)\right)\,.\label{eqn:randomv99v41}
\end{align}
For the first sum, we use $|d_i|\leq 1$, thus
\begin{align}
\sum_{i=1}^k(4k\bar q_i-4\bar Q_1)^2\var(\bar d_i|\mathcal F_t) \leq ck (k\bar Q_2- (\bar Q_1)^2) = ck \psi(t)\label{eqn:random19vj94} \,,
\end{align}
and for the second term,
\begin{align}
\ind_{G_{t+1}}\E\left(\left[(2k-2)\bar D_2\right]^2\given \mathcal F_t\right) &\leq 4k^2 \ind_{G_{t+1}}\E\left( \sum_{i=1}^k (\bar d_i)^4\given \mathcal F_t\right) \nonumber
\leq 4k^2 \ind_{G_{t+1}}\E\left( \sum_{i=1}^k |\bar d_i|\given \mathcal F_t\right)\nonumber\\
&\leq  \ind_{G_{t+1}}\frac{4k^2}{n^2}(Q_2+\bar Q_1 A)\leq \ind_{G_{t+1}}\frac{4k^2}{n^2}(Q_2+a A)\leq ck\label{eqn:random19vj92d94} \,,
\end{align}
where in the previous to last step we used that $\bar Q_1 = \sum_{i=1}^k a_i/w_i \leq a$, and in the last step that $q_i \leq c n/k$ and $a \leq cn/k$ in the event $\{\xi(t) \in \mathcal E\} \subseteq G_{t+1}$. By combining \cref{eqn:randomv99v41}, \cref{eqn:random19vj94} and  \cref{eqn:random19vj92d94} we obtain \cref{eqn:VarPot2}. 

Finally, for \cref{eqn:diffPot2},
\begin{align*}
|\ind_{G_{t+1}}\psi(t+1)-\E(\ind_{G_{t+1}}\psi(t+1)|\mathcal F_t)|&= \ind_{G_{t+1}}\left|\sum_{i=1}^k \sum_{j=1}^k\left[2\bar q_{ij}(\bar d_{ij}-\E(\bar d_{ij}|\mathcal F_t))+(\bar d_{ij})^2-\E((\bar d_{ij}^2)|\mathcal F_t) \right] \right|\nonumber\\
&\leq \ind_{G_{t+1}}\left(2\sqrt{\sum_{i=1}^k \sum_{j=1}^k \bar q_{ij}^2} \sqrt{\sum_{i=1}^k \sum_{j=1}^k(\bar d_{ij}-\E(\bar d_{ij}|\mathcal F_t))^2}+ \sum_{i=1}^k \sum_{j=1}^k(\bar d_{ij})^2+\E(\bar d_{ij}^2|\mathcal F_t)\right)\nonumber\\
&\leq c(\ind_{G_{t+1}}\phi(t)\sqrt{k}+k)\leq c(\ind_{G_{t}}\phi(t)\sqrt{k}+k) \,,
\end{align*}
where in the first inequality we used the Cauchy-Schwarz inequality, and in the second we used that at most one agents changes its colour and thus $|d_{ij}|\leq 1$\,.
\end{proof}

\subsubsection{Proof of \cref{lemma:firstPotentialSmall} and \cref{lemma:secondPotentialSmall}  }\label{sec:phase2PoteSmall}

The proofs of \cref{lemma:firstPotentialSmall} and \cref{lemma:secondPotentialSmall} are essentially the same as \cref{lemma:potential1Prop} and~\cref{lemma:potential2Prop}\,, showing that both protential functions satisfy the same properties (with different constants).

\begin{proof}[Proof of \cref{lemma:firstPotentialSmall}]

Set $T = \lfloor qwn \rfloor$ where $q$ is some constant that we will determine later. Let $C_1,\ldots, C_6$ be the constants in \cref{lemma:potential1Prop}.

Define the event $\mathcal B_j$ by
\begin{align*}
\mathcal B_{j} = \{\ind_{B_s}\phi(s) \leq n^2/2^j, \forall s\in\{jT,jT+1,\ldots, n^8\}\},
\end{align*}
and note that $\Prob(\mathcal B_0)=1$. Let $J = \min\{ j\geq 0: n^2/2^j \leq C\log(n)wn\}$. We will prove that $\mathcal B_{JT}$ holds with high probability.

For any $j\in\{1,\ldots, J\}$ we have that
\begin{align*}
\Prob((\mathcal B_{jT})^c)& \leq \Prob\left(\left\{ \exists s\in \{jT,\ldots, n^8\}: \ind_{B_s}\phi(s)> n^2/2^{j} \right\} \cap \mathcal B_{({j-1})T}\right) + \Prob((\mathcal B_{({j-1})T})^c)\,.
\end{align*}
For the first term on the right-hand side we have
\begin{align*}
\Prob\left(\left\{ \exists s\in \{jT,\ldots, n^8\}: \ind_{B_s}\phi(s)> n^2/2^{j} \right\} \cap \mathcal B_{({j-1})T}\right)\leq \sum_{s=jT}^{n^8}\Prob(\ind_{G_s}\phi(s)> n^2/2^{j})\,,
\end{align*}
where $G_s = B_s \cap \{\phi((j-1)T)\leq n^2/2^{j-1}\}$ for $s\geq (j-1)T$, and note that $G_{s+1} \in  \mathcal F_{s}$ for $s\geq (j-1)T$.

Now, we iterate \cref{eqn:ExPot1} of \cref{lemma:potential1Prop}, and choose $q$ large enough such that $(1-C_1/(nw))^{\lfloor qwn \rfloor } \leq \frac{1}{8}$\,. Then, letting $m = n^2/2^{j-1}$ we obtain for any $t \geq jT$\,,
\begin{align}
    \E(\ind_{G_t}\phi(t))&\leq \ind_{G_{(j-1)T}}\phi((j-1)T)(1-C_1/(nw))^T+C_2wn/C_1 \nonumber\\
    &\leq m/8+ \frac{C_2}{C_1}wn \leq m/8+ \frac{Cwn}{4} \leq m/4 \,.\label{eqn:random19s9zxcv}
\end{align}
In the last inequality we used the definition of $J$, and in particular that $m\geq n^2/2^{J-1}> C(\log n) w n$.

Now, define the process $M(t) = \ind_{G_{t+jT}}\phi(t+jT)$ for $t\geq 0$ and apply \cref{lemma:processConcentration} with 
\begin{align*}
    \alpha = C_1/(wn), \beta = C_2, \gamma = C_5k\sqrt{m}+C_6k, \delta^2 = C_3m/k+C_4, \text{ and }\lambda = m/4\,,
\end{align*}
then we obtain for $t\geq 0$\,,
\begin{align}
   \Prob\left(M(t)\geq m/2\right)\leq  \Prob\left(M(t)\geq \E(M(0))+m/4\right) &\leq \exp\left(-\frac{\lambda^2/2}{\frac{\delta^2}{(2\alpha-\alpha^2)}+\frac{\lambda \gamma}{3}}\right)\nonumber\\
    &\leq \exp\left(- c\min(m/n,m^{1/2}/k) \right) \,,\label{eqn:randomi952g}
\end{align}
where $c$ is a small constant that only depends on $C_1, C_2, C_3$ and $C_4$. Now let $r>0$, then by choosing $C$ large enough we get
\begin{align*}
\exp\left(- c\min(m/n,m^{1/2}/k) \right) \leq n^{-r}\,.
\end{align*}
Therefore, for all $s\geq jT$\,,
\begin{align*}
\Prob(\ind_{G_s}\phi(s)> n^2/2^{j}) \leq n^{-r}\,.
\end{align*}
Thus we conclude that for any $j\in\{1,\ldots, J\}$ we have that
\begin{align*}
\Prob((\mathcal B_{jT})^c)& \leq n^8n^{-r} + \Prob((\mathcal B_{j})^c) \leq Jn^8n^{-r} +\Prob(\mathcal B_0^c) = n^{10}n^{-r}\,,
\end{align*}
since $\Prob(\mathcal B_0)= 1$ and $J = O(\log n)$.
\end{proof}

As for the proof of \cref{lemma:secondPotentialSmall}\,, we follow the same steps but replacing $\ind_{B_t}\phi(t)$ by $\ind_{B_t'}\psi(t)$, and the definition of $J$ by $J = \min\{ j\geq 0: n^2/2^j \leq C'\log(n)wn\}$, with $C'>32C$ (where $C$ is the constant chosen in the definition of the set of configurations $\mathcal E'$ in \cref{eqn:defiE'}), and the constants $\alpha$, $\beta$, $\gamma$ and $\delta$ are replaced by
\begin{align*}
    \alpha = C_1/n, \beta = C_2, \gamma = C_5k\sqrt{m}+C_6k, \delta^2 = C_3km+C_4, \text{ and }\lambda = m/4\,.
\end{align*}
With the previous changes the RHS of \cref{eqn:randomi952g} is replaced by $\exp\left(- c\min(m/(kn),m^{1/2}/k) \right)$. Note, however, that $m/(kn) \geq C'\log n$, so we can still choose $C'$ large enough such that $\exp\left(- c\min(m/(kn),m^{1/2}/k) \right) \leq n^{-r}$.

\subsubsection{Proof of~\cref{lemma:processConcentration}}\label{sec:concentrationproof}
We give a proof of \cref{lemma:processConcentration}. Our proof follows a similar method as the one used in the proofs Theorems 7.3 and 7.5 of \cite{chung2006Concentration}\,.

Let $s\geq 0$ to be chosen later. A simple computation shows for $t \geq 1$\,,

\begin{align}
    \E(e^{sM(t)}|\mathcal F_{t-1}) &= e^{s\E(M(t)|\mathcal F_{t-1})} \E(e^{sM(t)-\E(sM(t)|\mathcal F_{t-1})}|\mathcal F_{t-1})\nonumber\\
    &= e^{s\E(M(t)|\mathcal F_{t-1})}\left(\sum_{j=0}^{\infty} \frac{s^j \E\left((M(t)-\E(M(t)|\mathcal F_{t-1})^j |\mathcal F_{t-1}\right)}{j!} \right) \,. \label{eqn:randombut1wx}
\end{align}

Let $G(t) = 2\sum_{j=2}^{\infty} \frac{t^{j-2}}{j!}$. For $t\leq 3$ we have that the function $G(t)$ is increasing and that $G(t) \leq 1/(1-t/3)$. Then
\begin{align}
    \sum_{j=0}^{\infty}\frac{s^j}{j!}\E&\left((M(t)- \E(M(t)|\mathcal F_{t-1})^j |\mathcal F_{t-1}\right) \\
    &= 1+ \sum_{j=2}^{\infty} \frac{s^j}{j!}\E\left((M(t)-\E(M(t)|\mathcal F_{t-1})^j |\mathcal F_{t-1}\right)\nonumber\\
    &\leq 1+ G(s\gamma)\frac{s^2}{2}\var(M(t)|\mathcal F_t)  \quad (\text{by~\cref{lemma:pM5ii}})\nonumber\\
    &\leq \exp\left({\frac{(s\delta)^2}{2} G(s\gamma)}\right) \quad (\text{by~\cref{lemma:pM5iii}}) \,.
\label{eqn:randomvru40g}
\end{align}

By replacing \cref{eqn:randomvru40g} into \cref{eqn:randombut1wx} and by \cref{lemma:pM5i} of the statement, we get
\begin{align*}
     \E(e^{sM(t)}|\mathcal F_{t-1}) &\leq e^{\E(sM(t)|\mathcal F_{t-1})}\left(e^{\frac{(s\delta)^2}{2}G(s\gamma)} \right) \nonumber\\
     &\leq \exp\left((1-\alpha)sM(t-1)+s\beta +\frac{(s\delta)^2}{2} G(s\delta)\right) \quad (\text{by~\cref{lemma:pM5i}})\,.
\end{align*}
Taking expectation on the above equation gives
\begin{align}
     \E e^{sM(t)} &\leq \left(\E e^{(1-\alpha)sM(t-1)}\right)\exp\left(s\beta+\frac{(s\delta)^2}{2} G(s\gamma)\right) \,.
\end{align}

By iterating and using that $G(x)$ is increasing for $x \leq s\gamma <3$ (we will choose $s$ to ensure $s\gamma<3$). 
\begin{align}
     \E e^{sM(t)} &\leq \left(\E e^{(1-\alpha)sM(t-1)}\right)\exp\left(s\beta+\frac{(s\delta)^2}{2} G(s\gamma)\right) \nonumber\\
     &\leq \exp\left((1-\alpha)^t M(0)s+ s\beta\sum_{i=0}^{t-1} (1-\alpha)^i \quad + \frac{(s\delta)^2G(s\gamma)}{2}\sum_{i=0}^{t-1} (1-\alpha)^{2i}\right)\,.
\end{align}

Then by Markov's inequality,
\begin{align*}
    \Prob(M(t)\geq \E(M(t))+\lambda) &\leq \exp\left({-\lambda s+\frac{(s\delta)^2G(s\gamma)}{2}\sum_{i=0}^{t-1} (1-\alpha)^{2i}}\right) \leq \exp\left({-\lambda s + \frac{(s\delta)^2G(s\gamma)}{2\alpha-\alpha^2}}\right)\,.
\end{align*}
Set $s = \frac{\lambda}{\delta^2/(2\alpha-\alpha^2)+\lambda \gamma/3}$, which ensures that $s\gamma \leq 3$ (i.e. we are working in the part where $G$ is increasing). Using that $G(x) \leq 1/(1-x/3)$ for $x \leq 3$, we get
\begin{align*}
 \Prob(M(t)\geq \E(M(t))+\lambda)\leq \exp\left( \frac{-\lambda^2/2}{\frac{\delta^2}{2\alpha-\alpha^2}+\lambda \gamma/3}\right)\,.
 \end{align*}

\subsection{Proof of \cref{sec:props}}\label{sec:proofFairness}
For the proof of \cref{lemma:thirdPotentialSmall} is rather similar to the ones of \cref{lemma:firstPotentialSmall} and \cref{lemma:secondPotentialSmall}. The proof is a corollary of the following lemma, which is the analogue of \cref{lemma:potential1Prop} and \cref{lemma:potential2Prop}  for the potential $\sigma(t) = (A(t)/w-a(t))$.

\begin{lemma}\label{lemma:potential3Prop}
Let $T\geq 0$ be an integer. Let $G_t$ be an event such that $\widehat{B}_{t}\subseteq G_{t+1} \in \mathcal F_{t}$ for all $t\geq T$, and such that $G_{t}\subseteq G_{t+1}$. Then there exist constants $C_1,C_2,C_3,C_4,C_5,C_6$ (independent of $T$, $k$, etc.) such that for any $t\geq T$  
\begin{enumerate}
    \item \label{eqn:ExpectationPot3}$
  \E(\ind_{G_{t+1}}\sigma^2(t+1)|\mathcal F_t)\leq  \ind_{G_{t}}\sigma^2(t)\left(1-\frac{C_1}{n}\right)+  C_2$
    \item \label{eqn:VarPot3}$
   \var(\ind_{G_{t+1}}\sigma^2(t+1)|\mathcal F_t) \leq C_3 \ind_{G_{t}}\sigma^2(t)+C_4.$
    \item \label{eqn:diffPot3}
    $|\ind_{G_{t+1}}\sigma^2(t+1)-\E(\ind_{G_{t+1}}\sigma^2(t+1)|\mathcal F_t)|\leq C_5|\sigma(t)|+C_6$ \,.
\end{enumerate}
\end{lemma}

With the above lemma, the proof of \cref{lemma:thirdPotentialSmall} follows exactly the same steps of proofs of \cref{lemma:firstPotentialSmall} and \cref{lemma:secondPotentialSmall}, and thus omitted.

\begin{proof}[Proof of \cref{lemma:potential3Prop}]
In the event $\widehat B_t$ we have that
\begin{align}
\sum_{i=1}^k\sum_{j=1}^k \left(\frac{A_i(t)}{w_i}-\frac{A_j(t)}{w_j}\right)^2 \leq \widehat Cwn\log n \,.\label{eqn:random8ruqqwez}
\end{align}

Rearranging the terms we have that
\begin{align}
\frac{1}{k^2}\sum_{i=1}^k\sum_{j=1}^k \left(\frac{A_i(t)}{w_i}-\frac{A_j(t)}{w_j}\right)^2 &= 2\left(\sum_{i=1}^k \frac{(A_i(t)/w_i)^2}{k}-\left(\sum_{i=1}^k \frac{1}{k}\frac{A_i(t)}{w_i} \right)^2\right)\nonumber\\
&= 2\left(\sum_{i=1}^k \frac{1}{k}\left(\frac{A_i(t)}{w_i}-X(t) \right)^2 \right) \leq \frac{\widehat Cw n\log n}{k^2}\,.
\end{align} 

Denote $X(t)= \frac{1}{k}\sum_{i=1}^k \frac{A_i(t)}{w_i}$, then we have that with probability
\begin{align*}
 \frac{1}{k}|A(t)-wX(t)| &\leq \frac{1}{k}\left|\sum_{i=1}^k A_i(t)-w_iX(t)\right|  = \frac{1}{k}\left|\sum_{i=1}^k \left(\frac{A_i(t)}{w_i}-X(t)\right)w_i\right| \nonumber\\
 &\leq \sqrt{\sum_{i=1}^k \frac{1}{k}\left(\frac{A_i(t)}{w_i}-X(t) \right)^2 }\sqrt{\frac{1}{k}\sum_{i=1}^k w_i^2} \leq c_1 \sqrt{n\log n}
\end{align*}
for some constant $c_1>0$ (this constant depends on $w$ and $k$, however, we have assumed they are constants). 

We also have that
\begin{align*}
\left(\sum_{i=1}^k \frac{A_i(t)^2}{w_i}\right) &-wX(t)^2 = \sum_{i=1}^k \left(\frac{A_i(t)^2}{w_i^2} -X(t)^2\right)w_i \nonumber\\
&\leq \sqrt{\sum_{i=1}^k \left(\frac{A_i(t)}{w_i}-X(t)\right)^2} \sqrt{\sum_{i=1}^k \left(\frac{A_i(t)}{w_i}+X(t)\right)^2 w_i^2} \leq c_1n^{3/2}\sqrt{\log n}
\end{align*}
since $A_i(t)\leq n$ and $X(t)\leq n$. Thus, combining the two inequalities we get
\begin{align}\label{eqn:random84g84jd00d}
\left|\sum_{i=1}^k \frac{A_i(t)^2}{w_i} - \frac{A(t)^2}{w}\right| \leq c_1n^{3/2}\sqrt{\log n}\,.
\end{align}

Now, denote $\sigma(t) = A(t)/w-a(t)$, then 
\begin{align*}
&\E\left( \ind_{G_{t+1}}\sigma(t+1)^2\given\mathcal F_t\right)\leq \ind_{G_{t+1}}\left(\sigma(t)^2 + 2\sigma(t)\left(1+\frac{1}{w} \right)\frac{1}{n(n-1)}\left(a(t)A(t)-\sum_{i=1}^k A_i(t)^2/w_i \right)\right) + 4.
\end{align*}

In the previous inequality we just wrote $\sigma(t+1) = \sigma(t) + (\sigma(t+1)-\sigma(t))$, and square both sides. Notice that $|(\sigma(t+1)-\sigma(t))|\leq (1+\frac{1}{w})$ as at most one agent changes its colour. Also, notice that the probability that a dark agent changes its colour to light at time $t$ is given by $\sum_{i=1}^k \frac{(A_i(t)^2/w)}{n(n-1)} $ and the probability it changes from light to dark is $\frac{a(t)A(t)}{n(n-1)}$. By using the bound given in \cref{eqn:random84g84jd00d}, we get
\begin{align*}
\E\left(\ind_{G_{t+1}} \sigma(t+1)^2\given\mathcal F_t\right)&\leq \ind_{G_{t+1}}\left(\sigma(t)^2 + c_2\frac{\sigma(t)}{n^2}\left(a(t)A(t)-\frac{A(t)^2}{w} \right)\right) + c_1\sqrt{\frac{\log n}{n}}+4\nonumber\\
&= \ind_{G_{t+1}}\left(\sigma(t)^2 + c_2\frac{\sigma(t)}{n^2}\left(a(t)A(t)-\frac{A(t)^2}{w} \right)\right) + 5\nonumber\\
&= \ind_{G_{t+1}}\sigma(t)^2\left(1-\frac{c_2}{n^2}A(t)\right)+ 5\nonumber\\
&= \ind_{G_{t+1}}\sigma(t)^2\left(1-\frac{c_2}{n}\right)+5
\end{align*}
where in the last step we used that in the event $\widehat B_t$ (which is contained in $G_{t+1}$) we have that $A(t) \geq c_3n$ for some constant $c_3 > 0$ (the constants $c_1, c_2$ etc.. may change value from line to line).

For the variance, notice that
\begin{align*}
\var(\ind_{G_{t+1}}\sigma(t+1)^2|\mathcal F_t)& = \var(\ind_{G_{t+1}}(\sigma(t)+(\sigma(t+1)-\sigma(t))^2|\mathcal F_t) \nonumber\\
&\leq 8\sigma(t)^2\ind_{G_{t+1}}\var(\sigma(t+1)-\sigma(t)|\mathcal F_t)+ 8\var(\ind_{G_{t+1}}(\sigma(t+1)-\sigma(t))^2|\mathcal F_t)\,,
\end{align*}
where in the inequality we used that $cov(X,Y) \leq 2(\var(X)+\var(Y))$ for any pair of random variables $X,Y$. Finally, note that  $|\sigma(t+1)-\sigma(t)|\leq 1+1/w\leq 2$ as at most only one agent changes its colour from a light one to dark one, or viceversa. Then
\begin{align*}
\var(\ind_{G_{t+1}}\sigma(t+1)^2|\mathcal F_t)\leq C_3\ind_{G_{t}}\sigma^2(t) + C_4 \,.
\end{align*}
Finally, since $|\sigma(t+1)-\sigma(t)|\leq 2$\,, we have
\begin{align*}
\E(\ind_{G_{t+1}}\sigma(t+1)^2|\mathcal F_t)-\ind_{G_{t+1}}\sigma(t+1)^2| &\leq \ind_{G_{t+1}}\left|\E(2\sigma(t)(\sigma(t+1)-\sigma(t))+(\sigma(t+1)-\sigma(t))^2|\mathcal F_t)\right|\nonumber\\
& \qquad + \ind_{G_{t+1}}\left|2\sigma(t)(\sigma(t+1)-\sigma(t))+(\sigma(t+1)-\sigma(t))^2\right|\nonumber\\
&\leq \ind_{G_{t}}C_5 |\sigma(t)|+ C_6 \,.
\end{align*}

\end{proof}

\appendix

\section{Auxiliary Results}

\begin{theorem}[Chapter XIV.2, XIV.3 in \cite{f68}]\label{pro:CaminataAleatoriaParcial}
Let $p \in (0,1)\setminus\{ 1/2\}$ and $b,s\in\mathbb{N}$. Consider a discrete time Markov chain $(Z_t)_{t\geq 0}$
with state space $\Omega=[0,b]$ where
\begin{itemize}
\item $Z_0=s \in [0,b]$
\item  $\Pr{Z_t=i ~|~ Z_{t-1}=i-1}=p$ for $i\in [1,b-1], t\geq 1$
\item $\Pr{Z_t=i ~|~ Z_{t-1}=i+1}=1-p$ for $i\in [1,b-1] , t\geq 1$
\item $\Pr{Z_t=i ~|~ Z_{t-1}=i}=1$ for $i\in \{0,b \}, t\geq 1$\,.
\end{itemize}
Let $T= \min \{ t\geq 0 ~|~ Z_t \in \{ 0,b \}\}$. Then,
\begin{align*}
\Pr{Z_T=b}=\frac{\left(\frac{1-p}{p}\right)^{s} - 1  }{ \left(\frac{1-p}{p}\right)^{b} - 1} \,,\;\text{ and }\; \Pr{Z_T=0}=\frac{   \left(\frac{1-p}{p}\right)^{b}- \left(\frac{1-p}{p}\right)^{s}   }{ \left(\frac{1-p}{p}\right)^{b} - 1}\,.
\end{align*}
Moreover,
\[ \Ex{T}=\frac{s}{1-2p} - \frac{b}{1-2p} \frac{1-\left(\frac{1-p}{p}\right)^s}{1-\left(\frac{1-p}{p}\right)^b}\,. \]
\end{theorem}

The following the Chernoff bound for Markov Chains is from \cite{chung2012}, however, we simplified their general statement for the purposes of our paper.

\begin{theorem}[\cite{chung2012}]\label{thm:MCchernoff}
Let $M$ be an ergodic Markov Chain with stationary distribution $\pi$ on a finite state space. Let $T_{mix}$ represent the $(1/8)$-mixing time. Let $N_i$ denote the number of hit of $M$ to the state $i$ in the first $t$ steps, then for $0<\delta<1$, it holds that
\begin{align*}
\Prob(|N_i-\pi(i)t|\leq \delta \pi(i)t)\leq c \exp\left(-\delta^2 \pi(i)t/(72T_{mix}) \right),
\end{align*}
where $c>0$ is a constant independent of $\delta$ and $\pi$, and $t$.
\end{theorem}

\begin{acks}
We would like to express our sincere gratitude to Colin Cooper and Tomasz Radzik, for their insightful suggestions and stimulating discussions. Nicolás Rivera was supported by the the Millennium
Institute for Foundational Research on Data (IMDF), and was supported by FONDECYT grant number 3210805.
\end{acks}

\bibliographystyle{ACM-Reference-Format}
\bibliography{references.bib}

\end{document}